\documentclass[10pt,a4paper]{article}

\usepackage{amsmath,amssymb,amsthm,mathrsfs}
\usepackage{stmaryrd}  %\interleave
\usepackage{graphicx}
\usepackage{xcolor}
\usepackage{fullpage}
\usepackage{bm}

\newlength{\defbaselineskip}
\newcommand{\setlinespacing}[1]%
{\setlength{\baselineskip}{#1 \defbaselineskip}}

\theoremstyle{plain}
\newtheorem{theorem}{Theorem}[section]
\newtheorem{lemma}[theorem]{Lemma}
\newtheorem{proposition}[theorem]{Proposition}
\newtheorem{corollary}[theorem]{Corollary}

\theoremstyle{definition}

\theoremstyle{remark}
\newtheorem{remark}[theorem]{Remark}

\numberwithin{equation}{section}

\DeclareMathOperator*{\esssup}{ess\,sup}

\allowdisplaybreaks
\begin{document}

\title{
Mean Field Portfolio Games with Consumption}
\author{Guanxing Fu\footnote{Department of Applied Mathematics, The Hong Kong Polytechnic University, Hung Hom, Kowloon, Hong Kong. Email: guanxing.fu@polyu.edu.hk. G. Fu's research is supported by The Hong Kong RGC (ECS No.25215122), NSFC Grant No.12101523, the Start-up Fund P0035348 from The Hong Kong Polytechnic University,  as well as the Research Centre for Quantitative Finance, The Hong Kong Polytechnic University (P0042708).} }

\maketitle

%% You can define extra breakpoints to count certain text regions.
%TC:break Abstract
\begin{abstract}
	We study mean field portfolio games with consumption. For general market parameters, we establish a one-to-one correspondence between Nash equilibria of the game and solutions to some FBSDE, which is proved to be equivalent to some BSDE. Our approach, which is general enough to cover power, exponential and log utilities,  relies on martingale optimality principle in  \cite{Cheridito2011,HIM-2005} and dynamic programming principle in \cite{ET-2015,FR-2011}.  When the market parameters do not depend on the Brownian paths, we get the unique Nash equilibrium in closed form. As a byproduct, when all market parameters are time-independent, we answer the question proposed in \cite{LS-2020}: the {\it strong equilibrium} obtained in \cite{LS-2020} is unique in the essentially bounded space. 
\end{abstract}

{\bf Keywords:}{ mean field game, portfolio game, consumption, martingale optimality principle}
%% You can use these special %TC: tags to ignore certain parts of the text.
%TC:ignore

\section{Introduction}
As a game-theoretic extension of classical optimal investment problems in \cite{Merton1971}, portfolio games have received substantial considerations in the financial mathematics literature in recent years. In a portfolio game, each player chooses her investment and/or consumption to maximize her utility induced by some risk preference criterion, by taking her competitors' decisions into consideration. 
The goal of the portfolio game is to search for a Nash equilibrium (NE), such that no one would like to change her strategy unilaterally. 
One way to model the interaction among players in the portfolio game is through the price equilibrium; however, it typically leads to tractability issue. Another way to model the interaction is through the relative performance: each player's utility is driven by her own wealth as well as the relative wealth to  her competitors. 

The study of portfolio games with relative performance concerns dates back to \cite{ET-2015}, where many player portfolio games with common stocks and trading constraint were studied: in the context of complete markets, the unique NE was obtained for general utility functions; in the context of incomplete markets, the unique NE was obtained for games with exponential utility functions, where the uniqueness result was proved by establishing an equivalent relation between each NE for the game and each solution to a multidimensional BSDE. \cite{FR-2011} examined similar games as \cite{ET-2015} with a different focus: \cite{FR-2011} constructed counterexamples where no NE exists, by proving that the corresponding multidimensional BSDE has no solution. In contrast to \cite{ET-2015,FR-2011} where all players traded common stocks, \cite{LZ-2019} studied portfolio games where each player traded a different but correlated stock. Assuming all market parameters to be constant, \cite{LZ-2019} obtained the unique constant NE by solving coupled HJB equations. In constrast to classical utility functions, \cite{Zari-2020,RP-2021} studied portfolio games with forward utilities. 
Recently, \cite{Fu-2021} studied portfolio games with general market parameters. A one-to-one correspondence between each NE and each solution to some FBSDE was established by dynamic programming principle (DPP) and martingale optimality principle (MOP), so that the portfolio game was solved by solving the FBSDE.  In \cite{Fu-2021}, we also obtained an asymptotic expansion result in powers of the competition parameter. 

All the aforementioned results do not incorporate consumption. The only results on portfolio games with consumption, to the best of our knowledge, are \cite{LS-2020} and \cite{RP-2021b}. Assuming all market parameters to be constant, \cite{LS-2020} obtained a unique NE, which was called {\it strong equilibrium}\footnote{By \cite[Definition 2.1 and Definition 3.1]{LS-2020}, a \textit{strong equilibrium} is the one with time-independent investment rate and continuous consumption rate. Moreover, both the investment rate and the consumption rate are adapted to the initial filtration.}. \cite{RP-2021b} examined a portfolio game with both investment and consumption under the framework of forward performance processes; the market parameters were also assumed to be constant. 

In this paper, we will study portfolio games with consumption under classical utility criteria, where market parameters are allowed to be time-dependent.
Assume there are $N$ risky assets in the market, with price dynamics of asset $i\in\{1,\cdots,N\}$ following
\begin{equation}\label{price-i-intro}
	dS^i_t= S^i_t\Big( h^i_t\,dt+\sigma^i_t\,W^i_t+ \sigma^{i0}_t\,dW^0_t		\Big),
\end{equation}
where $h^i$ is the return rate, $\sigma^i$ is the volatility corresponding to the idiosyncratic noise $W^i$, and $\sigma^{i0}$ is the volatility corresponding to the common noise $W^0$.
We further assume that player $i$ specializes in asset $i$. 
Let $X^i$ be the wealth process of player $i$, whose dynamics is given by
 \begin{equation}\label{wealth-power-N}
	\setlength{\abovedisplayskip}{3pt}
	\setlength{\belowdisplayskip}{3pt}
	dX^i_t=\pi^i_tX^i_t\Big( h^i_t\,dt+\sigma^i_t\,dW^i_t+\sigma^{i0}_t\,dW^0_t		\Big)-c^i_t X^i_t\,dt,\quad X^i_0=x^i,
\end{equation}
where $x^i$ is the initial wealth, $\pi^i$ is the investment rate and $c^i$ is the consumption rate. 
The risk preference for each player is described by a power utility function, i.e. given other players' strategies, player $i$ chooses the pair of investment rate and consumption rate $(\pi^i,c^i)$ to maximize the expected power utility induced by her terminal wealth and intermediate consumption:
 \begin{equation}\label{cost-power-N}
  \max_{\pi^i,c^i} \mathbb E\left[	\frac{1}{\gamma^i}\left(X^i_T (  \overline X^{-i}_T )^{-\theta^i}	\right)^{\gamma^i}+\int_0^T\frac{\alpha^i}{\gamma^i}\left(c_s^iX^i_s(\overline{cX}_s^{-i})^{-\theta^i}\right)^{\gamma^i}\,ds	\right].
 \end{equation}
 Here, $\overline X^{-i}= \left(	\Pi_{j\neq i} X^j	\right)^{\frac{1}{N-1}} $ and $\overline{cX}^{-i}=\left(\Pi_{j\neq i}c^jX^j\right)^{\frac{1}{N-1} }$ are the performance indices of player $i$, and the constant $\alpha^i$ describes the relative importance of the utility induced by consumption and terminal wealth.  

Although our $N$-player game \eqref{wealth-power-N}-\eqref{cost-power-N} is solvable, in this paper we will focus on the corresponding mean field game (MFG), in order to make the statement more concise. According to e.g. \cite{CD-2018,HMC-2006,LL-2007}, the corresponding MFG is:
\begin{equation}\label{model-MFG-power}
	\left\{ \begin{split}
		1.&~\textrm{Fix }(\mu,\nu)\textrm{ in some suitable space};\\
		2.&~\textrm{Solve the optimization problem: }\\
		&~J(\pi,c)=\mathbb E\left[    \frac{1}{\gamma} (X_T\mu^{-\theta}_T)^\gamma 	+\int_0^T\frac{\alpha}{\gamma}\left(	c_sX_s(\nu_s)^{-\theta}	\right)^{\gamma}\,ds	\right]\rightarrow\max \textrm{ over }(\pi,c)\\
		&~\textrm{such that } dX_t=\pi_tX_t(h_t\,dt+\sigma_t\,dW_t+\sigma^0_t\,dW^0_t)-c_tX_t\,dt,~X_0=x;\\
		3.&~\textrm{Search for the fixed point such that }\\
		&~  \mu_t=\exp\left(\mathbb E[\log X^*_t|\mathcal F^0_t]\right)\textrm{ and }\nu_t=\exp\left(\mathbb E[\log c^*_t|\mathcal F^0_t]+\mathbb E[\log X^*_t|\mathcal F^0_t]	\right),\quad t\in[0,T],\\
		&~\textrm{where }X^*\textrm{ and }c^*\textrm{ are the optimal wealth and optimal consumption rate from }2.
	\end{split}\right.
\end{equation}
\textbf{Our contributions.} Our paper has two contributions. First, under general market parameters, we provide a one-to-one correspondence between each NE of \eqref{model-MFG-power} and each solution to some FBSDE. The FBSDE is further proved to be equivalent to some BSDE, which completely characterizes the NE of the MFG \eqref{model-MFG-power}. Specifically, the optimal consumption rate and the optimal investment rate for the representative player are characterized by the $Y$-component and the $Z$-component of the BSDE, respectively. In order to establish the equivalence, we need to prove two sides. On the one hand, for each NE of the MFG \eqref{model-MFG-power}, we will prove that there exists an FBSDE such that the NE can be characterized by this FBSDE. On the other hand, for each solution to this FBSDE, we will prove that the solution corresponds to an NE of the MFG \eqref{model-MFG-power}. The former claim can be proved by DPP, and the latter claim can be proved by MOP. 
Second, when market parameters do not depend on the Brownian paths, we explicitly solve the BSDE characterizing the NE. In particular, we obtain the unique NE in closed form. The assumption on market parameters implies that the BSDE reduces to an ODE; the $Z$-component is zero. Thus, the optimal investment rate can be obtained, since it is completely characterized by the $Z$-component   following our first contribution. The optimal consumption rate is the unique solution to a Riccati equation, which is derived from the above ODE. We emphasize that the uniqueness result strongly relies on the one-to-one correspondence established in the first main contribution.

{\bf Connections with existing literature.}
From a methodology perspective, our paper shares similarities with \cite{Cheridito2011,ET-2015,FR-2011}. Specifically,
 \cite{Cheridito2011} used MOP to solve utility maximization problems with both investment and consumption. We apply a similar argument to our portfolio game and prove that each solution to some mean field FBSDE yields an NE of \eqref{model-MFG-power}. A key difference of the MOP in \cite{Cheridito2011} and in our paper is the choice of strategies. We claim that the strategies used in \cite{Cheridito2011} is not suitable to portfolio games; refer to \cite[Remark 2.2]{Fu-2021}. The DPP we use is adapted from \cite{ET-2015,FR-2011}, where all players traded common stocks and there was no consumption. This is a key step to prove the uniqueness result. Note that \cite{RP-2021b} also obtained a uniqueness result for portfolio games with consumption under forward utilities, using MOP implicitly. However, the uniqueness result in \cite{RP-2021b} was implied by the definition of, especially the (super)martingale properties of forward utilities. Thus, it does not imply the uniqueness result for games with classical utilities functions.
Admittedly, the most similar paper to the current one is \cite{Fu-2021}, where we studied mean field portfolio games with only investment in a general framework. We also used DPP and MOP to establish the one-to-one correspondence. The current paper can be considered as a continuation of \cite{Fu-2021}.
From a modeling perspective, our paper is similar to 
\cite{LS-2020}; when all market parameters become time-independent, our model reduces to the one in \cite{LS-2020}. Using a PDE approach, \cite{LS-2020} obtained a strong equilibrium, which was proved to be unique among all strong ones. By our more probabilistic approach, we conclude that the strong equilibrium obtained in \cite{LS-2020} is unique in the essentially bounded space. Thus, we answered the question proposed in \cite{LS-2020}.

The rest of the paper is organized as follows. After the introduction of notation, in Section \ref{sec:equivalence}, we establish an equivalent relationship between each solution to some mean field (F)BSDE and each NE of the MFG \eqref{model-MFG-power} with general market parameters. In Section \ref{sec:closed-form-NE}, assuming that the market parameters do not depend on the Brownian paths, we construct the unique NE in closed form. 

\textbf{Notation.} Let $(W,W^0)$ be a two dimensional Brownian motion, defined on a probability space $(\Omega,\mathbb P)$. Here, $W$ denotes the idiosyncratic noise for the representative player, and $W^0$ denotes the common noise for all players. Moreover, $\mathbb G=\{ \mathcal G_t, t\in[0,T]\}$ is assumed to be the augmented natural filtration of $(W,W^0)$. The augmented natural filtration of $W^0$ is denoted by $\mathbb F^0=\{\mathcal F^0_t,t\in[0,T]\}$. To allow for additional heterogeneity across players, we let $\mathcal A$ be a $\sigma$-algebra that is independent of $\mathbb G$. Let $\mathbb F=\{\mathcal F_t, t\in[0,T]	\}$ be the $
\sigma$-algebra generated by $\mathcal A$ and $\mathbb G$.

Denote $\textrm{Prog}(\Omega\times[0,T])$ the space of all stochastic processes that are $\mathcal F$-progressively measurable. For each $\eta\in \textrm{Prog}(\Omega\times[0,T])$, define $\|\eta\|_\infty=\esssup_{\omega\in\Omega,t\in[0,T]} |\eta_t(\omega)| $. Let $L^\infty$ be the space of all essentially bounded stochastic processes, i.e.,
\[
L^\infty=\{	\eta\in\textrm{Prog}(\Omega\times[0,T]): \|\eta\|_\infty<\infty	\}.
\] 
Define the BMO space under $\mathbb P$ as
\[
H^2_{BMO}=\left\{ \eta\in\textrm{Prog}(\Omega\times[0,T]):		\|\eta\|_{BMO}^2:= \sup_{\tau:\mathcal F-\textrm{stopping time}}\left\|\mathbb E^{}\left[ \left.	\int_\tau^T |\eta_t|^2\,dt	\right|\mathcal F_\tau	\right]\right\|_\infty	<\infty	\right\}.
\]
For each positive random variable $\xi$, denote $\widehat\xi:=\log(\xi)$.

Let $C$ be a generic positive constant, which may vary from line to line.

\section{Equivalence between the MFG \eqref{model-MFG-power} and Some (F)BSDE}\label{sec:equivalence}

We assume the space of admissible strategies for the representative player is $L^\infty\times L^\infty$.\footnote{The space $L^\infty$ to accommodate investment rates is smaller than $H^2_{BMO}$, which is commonly used in the literature. We use $L^\infty$ as our admissible space for technical purpose; see the estimates in the proof of Theorem \ref{thm:NE-FBSDE-power}. However, we do not lose much generality because the closed form investment rate constructed in Section \ref{sec:closed-form-NE} stays in $L^\infty$.}
Moreover, we say that the tuple $(\mu^*,\nu^*,\pi^*,c^*)$ is an NE of the MFG \eqref{model-MFG-power}, if $(\pi^*,c^*)$ is admissible, with $X^*$ being the corresponding wealth, $\mu^*_t=\exp\left(\mathbb E[\log X^*_t|\mathcal F^0_t]\right)$ and $\nu^*_t=\exp\left(	\mathbb E[\log(c^*_tX^*_t)|\mathcal F^0_t]	\right)$, $t\in[0,T]$, and if the optimality condition holds for each admissible strategy $(\pi,c)$: 
$$\mathbb E\left[    \frac{1}{\gamma} (X^*_T(\mu^*_T)^{-\theta})^\gamma 	+\int_0^T\frac{\alpha}{\gamma}\left(	c^*_sX^*_s(\nu^*_s)^{-\theta}	\right)^{\gamma}\,ds	\right]\geq \mathbb E\left[    \frac{1}{\gamma} (X_T(\mu^*_T)^{-\theta})^\gamma 	+\int_0^T\frac{\alpha}{\gamma}\left(	c_sX_s(\nu^*_s)^{-\theta}	\right)^{\gamma}\,ds	\right].$$ %In particular, $(\mu^*_t)_{0\leq t\leq T}$ is called a solution of \eqref{model-MFG-power}.

Throughout the paper, the following assumptions are in force. 

\textbf{Assumption 1.} 
The initial wealth $x$, risk aversion parameter $\gamma$, competition parameter $\theta$, and weight parameter $\alpha$ of the population are assumed to be bounded  $\mathcal A$-random variables. Moreover, $x$ and $\alpha$ are $\mathbb R_+$-valued, $\gamma$ is valued in $(-\infty,1)/\{ 0\}$ and $\theta$ is valued in $[0,	1]$.

Assume the return rate $h\in L^\infty$ and the volatilities $(\sigma,\sigma^0)\in L^\infty\times L^\infty$. Moreover, $|\gamma|$ and $|\sigma|+|\sigma^0|$ are bounded away from $0$, i.e., there exist positive constants  $\underline\gamma$ and $\underline\sigma$ such that  $|\gamma|\geq \underline\gamma>0$ a.s. and $|\sigma|+|\sigma^0|\geq \underline\sigma>0$ a.s. a.e..

\subsection{MFGs and Mean Field FBSDEs Are Equivalent}\label{sec: characterization}
In this section, we prove that the solvability of a mean field FBSDE is sufficient and necessary for the existence of an NE of the MFG \eqref{model-MFG-power}. First, 
using MOP as in \cite{Cheridito2011,HIM-2005}, we prove that the value function and optimal strategy of the associated optimization problem in \eqref{model-MFG-power} have a one-to-one correspondence to the solution of a BSDE. 
\begin{proposition}\label{prop:BSDE-optim}
	For fixed $(\mu,\nu)$, the value function $V$ and the unique optimal strategy $(\pi^*,c^*)$ of the associated optimization problem in \eqref{model-MFG-power} are given by
	\begin{equation}\label{optimal-strategy-optim}
		V = \frac{1}{\gamma}e^{\gamma\widehat X_0+Y_0},\qquad \pi^*_t= 	\frac{  h_t+\sigma_t  Z_t+\sigma^0_t  Z^0_t		}{(1-\gamma)(\sigma^2_t+(\sigma^0_t)^2)},\quad c^*_t=\alpha^{\frac{1}{1-\gamma}} e^{ -\frac{1}{1-\gamma}(\theta\gamma\widehat\nu_t+Y_t)  },\quad t\in[0,T],
	\end{equation}
	where $(Y,Z,Z^0)$ satisfies the following BSDE
	\begin{equation}\label{BSDE-optim}
		\left\{\begin{split}
				-dY_t=&~\Bigg\{\frac{Z^2_t+(Z^0_t)^2}{2}+\frac{\gamma}{2(1-\gamma)}\frac{ (h_t+\sigma_t Z_t+\sigma^0_tZ^0_t)^2    }{\sigma^2_t+(\sigma^0_t)^2}+(1-\gamma)\big\{\alpha e^{-(Y_t+\theta\gamma\widehat\nu_t)}\big\}^{\frac{1}{1-\gamma}} \Bigg\}\,dt\\
	&~-Z_t\,dW_t-Z_t^0\,dW^0_t,\\
	\widehat Y_T=&~-\gamma\theta\widehat\mu_T.
	\end{split}\right.
	\end{equation}
Here, we recall that $\widehat X_0=\log X_0$, $\widehat\mu_t=\log\mu_t$ and $\widehat\nu_t=\log\nu_t$, $t\in[0,T]$. 
\end{proposition}
\begin{proof}
	The proof is a modification of that in \cite{Cheridito2011}. As discussed in the introduction, the essential difference lies in the choice of strategies; our choice is appropriate to the game-theoretic model in this paper.
	
	Let $(Y,Z,Z^0)$ be a solution to \eqref{BSDE-optim}. We will prove that $(\pi^*,c^*)$ defined in \eqref{optimal-strategy-optim} is an optimal strategy of the associated optimization problem in \eqref{model-MFG-power}, by fixing $(\mu,\nu)$. To do so, for each strategy $(\pi,c)\in L^\infty\times L^\infty$, define $$R^{\pi,c}_t=\frac{1}{\gamma}  e^{\gamma\widehat X^{\pi,c}_t+Y_t}+\int_0^t\frac{\alpha}{\gamma}e^{\gamma\widehat c_s+\gamma\widehat X^{\pi,c}_s-\theta\gamma\widehat\nu^*_s}\,ds.$$ 
	We will prove that $R^{\pi,c}$ satisfies the following three items:
	\begin{equation}\label{claim-MOP-optim}
		\begin{split}
			1. &\quad 
			\textrm{for any }(\pi,c), ~R^{\pi,c}\textrm{ is a supermartingale};\\
			2. &\quad R^{\pi^*,c^*} \textrm{ is a martingale for } (\pi^*,c^*) \textrm{ in }\eqref{optimal-strategy-optim};\\
			3. &\quad R^{\pi,c}_0 \textrm{ is independent of }(\pi,c).
		\end{split}
	\end{equation}
	If the claim \eqref{claim-MOP-optim} is true, it holds that $$\mathbb E[R^{\pi,c}_T]\leq R_0^{\pi,c}=R_0^{\pi^*,c^*}=\mathbb E[R^{\pi^*,c^*}_T].$$
	Thus,
	$(\pi^*,c^*)$ is optimal. %Moreover, the second equation and the last equation in \eqref{FBSDE-characterization} imply that

		It remains to prove the claim \eqref{claim-MOP-optim}. Denote by $f(Y,Z,Z^0)$ the driver of $Y$ in \eqref{BSDE-optim}. Applying It\^o's formula to $R^{\pi,c}$, we get
	\begin{equation*}
		\begin{split}
			dR^{\pi,c}_t%=&~\frac{1}{\gamma}X^\gamma_t e^{Y_t}\Bigg\{	 \bigg(		 \gamma\pi_sh_s-\gamma c_s-\frac{\gamma\pi^2_s}{2}(\sigma^2_s+(\sigma^0_s)^2	 )	+f_s(Z_s,Z^0_s)	 + \frac{1}{2}(\gamma\pi_s\sigma_s+Z_s)^2+\frac{1}{2}(\gamma\pi_s\sigma^0_s+Z^0_s )^2 \\
			%		&~\qquad \qquad \quad +\alpha e^{-Y_t}( c_s(\nu_s^*)^{-\theta} )^\gamma				\bigg) \,dt
			%		+(\gamma\pi_s\sigma_s+Z_s)\,dW_t+(\gamma\pi_s\sigma^0_s+Z^0_s)\,dW^0_t			\Bigg\}\\
			%=&~X^\gamma_t e^{Y_t}\Bigg\{	 \bigg(		  \pi_th_t-  c_t-\frac{ \pi^2_t}{2}(\sigma^2_t+(\sigma^0_t)^2	 )	+\frac{f_t(Z_t,Z^0_t)}{\gamma}	 + \frac{1}{2\gamma}(\gamma\pi_t\sigma_t+Z_t)^2+\frac{1}{2\gamma}(\gamma\pi_t\sigma^0_t+Z^0_t )^2 \\
			%	&~\qquad \qquad \quad +\frac{\alpha}{\gamma} e^{-Y_t}( c_t(\nu_t^*)^{-\theta} )^\gamma				\bigg) \,dt
			%	+\frac{\gamma\pi_t\sigma_t+Z_t}{\gamma}\,dW_t+\frac{\gamma\pi_t\sigma^0_t+Z^0_t}{\gamma}\,dW^0_t			\Bigg\}\\
			=&~X^\gamma_te^{Y_t} \Bigg\{ \bigg(	-\frac{1-\gamma}{2}(\sigma^2_t+(\sigma^0_t)^2)\pi^2_t+(h_t+\sigma_tZ_t+\sigma^0_tZ^0_t)\pi_t + \frac{Z^2_t+(Z^0_t)^2}{2\gamma}	-\frac{f_t(Y_t,Z_t,Z^0_t)}{\gamma}\bigg)\,dt	\\
			&~\qquad\qquad+\bigg(	-c_t+\frac{\alpha}{\gamma} e^{-Y_t}(c_t(\nu^*_t)^{-\theta})^\gamma		\bigg)\,dt		+\frac{\gamma\pi_t\sigma_t+Z_t}{\gamma}\,dW_t+\frac{\gamma\pi_t\sigma^0_t+Z^0_t}{\gamma}\,dW^0_t					\Bigg\}\\
			=&~X_t^\gamma e^{Y_t} \Bigg\{		\bigg(	-\frac{1-\gamma}{2}\left(	\sigma^2_t+(\sigma^0_t)^2		\right)\bigg\{	 \pi_t-\frac{	h_t+\sigma_t Z_t+\sigma^0_tZ^0_t	}{(1-\gamma)(	\sigma^2_t+(\sigma^0_t)^2	)}			\bigg\}^2\bigg)\,dt	\\
			&~\qquad\qquad +  \bigg(-c_t+\frac{\alpha}{\gamma} e^{-Y_t}(c_t(\nu^*_t)^{-\theta})^\gamma	-\frac{1-\gamma}{\gamma}\left\{ \alpha e^{-Y_t}(\nu^*_t)^{-\theta\gamma}\right	\}^{\frac{1}{1-\gamma}} \bigg)\,dt\\
			%		&~\qquad \qquad +\frac{	(h+\sigma Z+\sigma^0Z^0)^2		}{2(1-\gamma)(	 \sigma^2+(\sigma^0)^2	)}+ \frac{Z^2_t+(Z^0_t)^2}{2\gamma}+\frac{1-\gamma}{\gamma}\left\{ \alpha e^{-Y_t}(\nu^*_t)^{-\theta\gamma}\right	\}^{\frac{1}{1-\gamma}}	+\frac{f_t(Z_t,Z^0_t)}{\gamma}		\bigg)		\,dt	\\
			&~\qquad\qquad +\frac{\gamma\pi_t\sigma_t+Z_t}{\gamma}\,dW_t+\frac{\gamma\pi_t\sigma^0_t+Z^0_t}{\gamma}\,dW^0_t	\Bigg\}.
		\end{split}
	\end{equation*}
	%By choosing
	%\[
	%f(Z,Z^0)=-\frac{Z^2+(Z^0)^2}{2}-\frac{\gamma (h+\sigma Z+\sigma^0Z^0)^2    }{2(1-\gamma)(\sigma^2+(\sigma^0)^2)}-(1-\gamma)\left(	\alpha e^{-Y}(\nu^*)^{-\theta\gamma}		\right)^{\frac{1}{1-\gamma}},
	%\]
	%it holds that
	%\begin{equation*}
	%	\begin{split}
	%		dR^{\pi,c}_t
	%		=&~X_t^\gamma e^{Y_t} \Bigg\{		\bigg(	-\frac{1-\gamma}{2}\left(	\sigma^2+(\sigma^0)^2		\right)\bigg\{	 \pi-\frac{	h+\sigma Z+\sigma^0Z^0	}{(1-\gamma)(	\sigma^2+(\sigma^0)^2	)}			\bigg\}^2\bigg)\,dt	\\
	%		&~\qquad\qquad +  \bigg(-c_t+\frac{\alpha}{\gamma} e^{-Y_t}(c_t(\nu^*_t)^{-\theta})^\gamma	-\frac{1-\gamma}{\gamma}\left\{ \alpha e^{-Y_t}(\nu^*_t)^{-\theta\gamma}\right	\}^{\frac{1}{1-\gamma}} \bigg)\,dt\\
	%		&~\qquad\qquad +\frac{\gamma\pi_t\sigma_t+Z_t}{\gamma}\,dW_t+\frac{\gamma\pi_t\sigma^0_t+Z^0_t}{\gamma}\,dW^0_t	\Bigg\}.
	%	\end{split}
	%\end{equation*}
	Note that for all $(\pi,c)$ the drift of $R^{\pi,c}$ is non-positive, and the drift of $R^{\pi^*,c^*}$ is zero for $(\pi^*,c^*)$ in \eqref{optimal-strategy-optim}.
	%\[
	%-c+\frac{\alpha}{\gamma} e^{-Y}(c(\nu^*)^{-\theta})^\gamma	-\frac{1-\gamma}{\gamma}\left\{ \alpha e^{-Y}(\nu^*)^{-\theta\gamma}\right	\}^{\frac{1}{1-\gamma}}	\leq 0,
	%\]
	Thus, the claim \eqref{claim-MOP-optim} is proved.
	
	Since $(\pi,c)\mapsto J(\pi,c)$ is concave, $(\pi^*,c^*)$ in \eqref{optimal-strategy-optim} is unique.
\end{proof}
The following theorem, which is the main result of this section, establishes the necessary and sufficient conditions for the solvability of our MFG \eqref{model-MFG-power}. The sufficient condition is a corollary of Proposition \ref{prop:BSDE-optim}.
%	\begin{equation}\label{FBSDE-power}
%	\left\{\begin{split}
%		%	d\widehat X_t=&~\Bigg\{	\frac{(	h_t+\sigma_tZ_t+\sigma^0_tZ^0_t)h_t	}{(1-\gamma)( \sigma^2_t+(\sigma^0_t)^2 )}		-\frac{(h_t+\sigma_tZ_t+\sigma^0_t Z^0_t)^2}{2 (1-\gamma)^2(\sigma^2_t+(\sigma^0_t)^2 )		}	 -\left\{\alpha e^{-Y_t}(\nu^*_t)^{-\theta\gamma}\right\}^{\frac{1}{1-\gamma}}	\Bigg\}\,dt	\\
%		d\widehat X^*_t=&~\bigg\{ \pi^*_t h_t-c^*_t-\frac{1}{2}(\pi^*_t)^2(\sigma^2_t+(\sigma^0_t)^2)\bigg\}\,dt+ \pi^*_t\sigma_t\,dW_t+\pi^*_t\sigma^0_t\,dW^0_t\\
%		%		&~	+\frac{( h_t+\sigma_t Z_t+\sigma^0_tZ^0_t )\sigma_t}{(1-\gamma)(\sigma^2_t+(\sigma^0_t)^2)}\,dW_t+\frac{(h_t+\sigma_tZ_t+\sigma^0_tZ^0_t)\sigma^0_t}{(1-\gamma)(\sigma^2_t+(\sigma^0_t)^2)}\,dW^0_t,			\\
%		-dY_t=&~\Bigg\{\frac{Z^2_t+(Z^0_t)^2}{2}+\frac{\gamma}{2(1-\gamma)}\frac{ (h_t+\sigma_t Z_t+\sigma^0_tZ^0_t)^2    }{\sigma^2_t+(\sigma^0_t)^2}+(1-\gamma)\big\{\alpha e^{-Y_t}(\nu^*_t)^{-\theta\gamma}\big\}^{\frac{1}{1-\gamma}} \Bigg\}\,dt\\
%		&~-Z_t\,dW_t-Z_t^0\,dW^0_t,\\
%		\widehat X^*_0=&~\log(x),~Y_T=-\gamma\theta\widehat\mu^*_T,\\
%		\widehat\mu^*_t=&~\mathbb E[\widehat X^*_t|\mathcal F^0_t],\quad \widehat\nu^*_t= \mathbb E[ \widehat{c^*_tX^*_t}|\mathcal F^0_t	].
%	\end{split}\right.
%\end{equation}
 In order to prove the necessary part, we rely on the dynamic programming principle as in \cite[Lemma 4.4]{ET-2015} and \cite[Lemma 3.2]{FR-2011}, where the $N$-player game with exponential utility functions and trading constraint but without individual noise was considered. In the next theorem, we adapt the argument to our MFG \eqref{model-MFG-power}. 
%\subsection{NE of MFG \eqref{model-MFG}}

\begin{theorem}\label{thm:NE-FBSDE-power}
$\bm{(1)}$ Let $(\mu^*,\nu^*,\pi^*,c^*) $ be an NE of the MFG \eqref{model-MFG-power}, such that 
	\begin{equation}\label{necessary:reverse-power}
		\mathbb E\left[	\frac{1}{\gamma} e^{\gamma (\widehat X^{*}_T-\theta\widehat\mu^*_T)}+\int_0^T\frac{\alpha}{\gamma}e^{\gamma\widehat c^*_s+\gamma\widehat X^*_s-\theta\gamma\widehat\nu^*_s}\,ds \Big|\mathcal F_\cdot	\right] \textrm{ satisfies }R_p\textrm{ for some }p>1.\footnote{The reverse H\"older inequality $R_p$ is defined in Appendix \ref{app:reverse}.}
	\end{equation}	
Then this NE must satisfy for $t\in[0,T]$
 	\begin{equation}\label{FBSDE-characterization}
 	\left\{\begin{split}
 		\widehat\mu^*_t=&~\mathbb E[\widehat X^*_t|\mathcal F^0_t],\\
 		\widehat\nu^*_t=&~\frac{\mathbb E[\widehat X^*_t|\mathcal F^0_t]}{ 1+\mathbb E\left[	\frac{\theta\gamma}{1-\gamma}	\right] }+\frac{\mathbb E\left[ \frac{\log\alpha}{1-\gamma} \right]}{ 1+\mathbb E\left[  \frac{\theta\gamma}{1-\gamma}	\right] } -\frac{\mathbb E\left[ \frac{Y_t}{1-\gamma}	|\mathcal F^0_t\right]}{1+\mathbb E\left[ \frac{\theta\gamma}{1-\gamma} \right]},\\
 		\pi^*_t=&~\frac{h_t+\sigma_t Z_t+\sigma^0_tZ^0_t}{(1-\gamma)(  \sigma^2_t+(\sigma^0_t)^2  )},\\
 		c^*_t=&~\left(	\alpha e^{-Y_t}(\nu^*_t)^{-\theta\gamma}	\right)^{\frac{1}{1-\gamma}},
 	\end{split}\right.
 \end{equation}
where $(\widehat X^*,Y,Z,Z^0)$ satisfies the following mean field FBSDE  
	\begin{equation}\label{FBSDE-power}
		\left\{\begin{split}
		d\widehat X^*_t=&~\bigg\{ \pi^*_t h_t-c^*_t-\frac{1}{2}(\pi^*_t)^2(\sigma^2_t+(\sigma^0_t)^2)\bigg\}\,dt+ \pi^*_t\sigma_t\,dW_t+\pi^*_t\sigma^0_t\,dW^0_t,\\
			-dY_t=&~\Bigg\{\frac{Z^2_t+(Z^0_t)^2}{2}+\frac{\gamma}{2(1-\gamma)}\frac{ (h_t+\sigma_t Z_t+\sigma^0_tZ^0_t)^2    }{\sigma^2_t+(\sigma^0_t)^2}+(1-\gamma)\big\{\alpha e^{-Y_t}(\nu^*_t)^{-\theta\gamma}\big\}^{\frac{1}{1-\gamma}} \Bigg\}\,dt\\
			&~-Z_t\,dW_t-Z_t^0\,dW^0_t,\\
			\widehat X^*_0=&~\log(x),~Y_T=-\gamma\theta\widehat\mu^*_T.
			\end{split}\right.
		\end{equation}
Here, we recall $\widehat\mu^*_t=\log\mu^*_t$, $\widehat\nu^*_t=\log\nu^*_t$, $\widehat c^*_t=\log c^*_t$ and $\widehat X^*_t=\log X^*_t$, $t\in[0,T]$.

$\bm{(2)}$ If there exists a solution to the FBSDE \eqref{FBSDE-power} with $(\mu^*,\nu^*,\pi^*,c^*)$ defined in \eqref{FBSDE-characterization} such that $(Z,Z^0)\in L^\infty\times L^\infty$ and $Y+\theta\gamma\mathbb E[\widehat X^*|\mathcal F^0]\in L^\infty$, then $(\mu^*,\nu^*,\pi^*,c^*)$ in \eqref{FBSDE-characterization} is an NE of the MFG \eqref{model-MFG-power}  
such that \eqref{necessary:reverse-power} holds. 
	
%	The relationship is given by $$\pi^*=\frac{h+\sigma Z+\sigma^0Z^0}{(1-\gamma)(  \sigma^2+(\sigma^0)^2  )}$$ and 
%	\[
%		c^*=\exp\left\{ \frac{\log\alpha}{1-\gamma} -\frac{Y+\theta\gamma\mathbb E[\widehat X|\mathcal F^0]}{1-\gamma} -\frac{\theta\gamma}{1-\gamma}\frac{  \mathbb E\left[	\frac{\log\alpha}{1-\gamma}\right]-\mathbb E\left[	\frac{ Y+\theta\gamma\mathbb E[\widehat X|\mathcal F^0]   }{1-\gamma}		|\mathcal F^0	\right] }{1+\mathbb E\left[ \frac{\theta\gamma}{1-\gamma} \right]} \right\}.
%	\]
\end{theorem}
\begin{proof}
$\bm{(1)}$ Let $(\mu^*,\nu^*,\pi^*,c^*)$ be an NE of \eqref{model-MFG-power} such that \eqref{necessary:reverse-power} holds. For each $(\pi,c)\in L^\infty\times L^\infty$, define
\begin{equation*}
	\begin{split}
M^{\pi,c}_t=&~e^{\gamma\widehat X^{\pi,c}_t} \esssup_{(\kappa,b)\in L^\infty\times L^\infty} \mathbb E\left[ \frac{1}{\gamma}e^{\gamma(\widehat X_T^{\kappa,b}-\widehat X^{\kappa,b}_t-\theta\widehat\mu^*_T)} +\int_t^T\frac{\alpha}{\gamma}e^{\gamma\widehat b_s+\gamma(\widehat X^{\kappa,b}_s-\widehat X_t^{\kappa,b})-\theta\gamma\widehat\nu^*_s}\,ds	\Big|\mathcal F_t		\right]\\
&~ + \int_0^t\frac{\alpha}{\gamma}e^{\gamma\widehat c_s+\gamma\widehat X^{\pi,c}_s-\theta\gamma\widehat\nu^*_s}\,ds,
	\end{split}
\end{equation*}
where $X^{\kappa,b}$ denotes the wealth process associated with the investment-consumption pair $(\kappa,b)$. 
Following the argument in \cite[Lemma 4.4]{ET-2015} and \cite[Lemma 3.2]{FR-2011}, $M^{\pi,c}$ has a continuous version which is a supermartingale for all $(\pi,c)$ and a martingale for $(\pi^*,c^*)$. Denote $\widehat X^*:=\widehat X^{\pi^*,c^*}$ and $ M^*:= M^{\pi^*,c^*}$.
Our goal is to get an SDE for $M^{\pi,c}$, and by the supermartingale property of $M^{\pi,c}$ and martingale property of $M^{*}$ we link $(\pi^*,c^*)$ to some FBSDE. We will achieve the goal by the following steps.

\textbf{Step 1: representation of $M^*$ and $M^{\pi,c}$.}
Note that $M^{*}_\cdot-\int_0^\cdot \frac{\alpha}{\gamma}e^{\gamma\widehat c^*_s+\gamma\widehat X^{*}_s-\theta\gamma\widehat\nu^*_s}\,ds  \neq 0$ a.s.. Thus, martingale representation theorem yields $(\breve Z,\breve Z^0)$ such that
\begin{equation}\label{martingale-representation}
	dM^{*}_t=\left(	  	M^{ *}_t	- \int_0^t \frac{\alpha}{\gamma}e^{\gamma\widehat c^*_s+\gamma\widehat X^{*}_s-\theta\gamma\widehat\nu^*_s}\,ds				\right)(\breve Z_t\,dW_t+\breve Z^0_t\,dW^0_t).
\end{equation}
Moreover, since $M^*$ is a positive martingale, \eqref{necessary:reverse-power} and Lemma \ref{lem:reverse} yield $(\mathring Z,\mathring Z^0)\in H^2_{BMO}\times H^2_{BMO}$ such that
\begin{equation}\label{representation-M*-2}
dM^{*}_t=M^{*}_t\big\{	\mathring Z_t\,dW_t+\mathring Z^0_t\,dW^0_t  	\big\},
\end{equation}
which together with \eqref{martingale-representation} implies that 
\begin{equation}\label{relation-Z-M*}
(\breve Z_t,\breve Z^0_t)=\frac{M^{*}_t}{ 	  	M^{*}_t	- \int_0^t \frac{\alpha}{\gamma}e^{\gamma\widehat c^*_s+\gamma\widehat X^{*}_s-\theta\gamma\widehat\nu^*_s}\,ds				 }( \mathring Z_t,\mathring Z^0_t  ),\qquad a.s.~\omega\in\Omega,~~a.e.~ t\in[0,T].
\end{equation}
By the definition of $M^{\pi,c}$ and $M^{ *}$, we have
\begin{equation}\label{expression-M}
	\begin{split}
		M^{\pi,c}_t=&~e^{\gamma( \widehat X^{\pi,c}_t - \widehat X^{*}_t )}\left(M^{ *}_t	 -\int_0^t \frac{\alpha}{\gamma}e^{\gamma\widehat c^*_s+\gamma\widehat X^{ *}_s-\theta\gamma\widehat\nu^*_s}\,ds			\right)+\int_0^t  \frac{\alpha}{\gamma}e^{\gamma\widehat c_s+\gamma\widehat X^{\pi,c}_s-\theta\gamma\widehat\nu^*_s}\,ds.
	\end{split}
\end{equation}
Finally, in this step we define a stochastic process $Y$ for later use, which will turn out to be the backward component of the desired FBSDE:
\begin{equation}\label{def:Y-1}
	\begin{split}
		e^{Y_t}=\gamma e^{ - \gamma\widehat X^{ *}_t }\left(M^{ *}_t	 -\int_0^t \frac{\alpha}{\gamma}e^{\gamma\widehat c^*_s+\gamma\widehat X^{ *}_s-\theta\gamma\widehat\nu^*_s}\,ds			\right).
	\end{split}
\end{equation}
Using \eqref{expression-M}, it also holds that
\begin{equation}\label{def:Y-2}
	e^{Y_t}=\gamma e^{ - \gamma\widehat X^{\pi,c}_t }\left(M^{\pi,c}_t	 -\int_0^t \frac{\alpha}{\gamma}e^{\gamma\widehat c_s+\gamma\widehat X^{\pi,c}_s-\theta\gamma\widehat\nu^*_s}\,ds			\right).
\end{equation}

\textbf{Step 2: SDE for $M^{\pi,c}$.} Recall $\widehat X^{\pi,c}$ and $\widehat X^*$ are the log-wealth associated with $(\pi,c)$ and $(\pi^*,c^*)$, respectively. It\^o's formula implies that
\begin{equation}\label{eq:to-supermartingale-1}
	\begin{split}
		&~de^{\gamma( \widehat X^{\pi,c}_t-\widehat X^{*}_t )}\\
		=&~ e^{\gamma( \widehat X^{\pi,c}_t-\widehat X^{*}_t )}\bigg\{	\gamma (\pi_t-\pi^*_t)h_t-\gamma(c_t-c^*_t)-\frac{\gamma}{2}(\pi^2_t-(\pi^*_t)^2)(\sigma^2_t+(\sigma^0_t)^2)\\
		&~	+\frac{\gamma^2}{2} (\pi_t-\pi^*_t)^2(	\sigma^2_t+(\sigma^0_t)^2	)  \bigg\}\,dt+e^{\gamma( \widehat X^{\pi,c}_t-\widehat X^{*}_t )}\bigg\{\gamma\sigma_t(\pi_t-\pi^*_t)\,dW_t+\gamma\sigma^0_t(\pi_t-\pi^*_t)\,dW^0_t\bigg\}.
	\end{split}
\end{equation}
From the expressions \eqref{expression-M} and \eqref{martingale-representation}, integration by parts implies that
\begin{align*}
%	\begin{split}
		&~dM^{\pi,c}_t\\
		=&~\left(M^{ *}_t	 -\int_0^t \frac{\alpha}{\gamma}e^{\gamma\widehat c^*_s+\gamma\widehat X^{ *}_s-\theta\gamma\widehat\nu^*_s}\,ds			\right)\,de^{\gamma( \widehat X^{\pi,c}_t -\widehat X^{ *}_t)}\\
		&~+e^{\gamma( \widehat X^{\pi,c}_t - \widehat X^{ *}_t )}\,dM^{ *}_t+\left(	-e^{\gamma( \widehat X^{\pi,c}_t -\widehat X^{ *}_t )}  \frac{\alpha}{\gamma}e^{\gamma\widehat c^*_t+\gamma\widehat X^{ *}_t-\theta\gamma\widehat\nu^*_t}  +  \frac{\alpha}{\gamma}e^{\gamma\widehat c_t+\gamma\widehat X^{\pi,c}_t-\theta\gamma\widehat\nu^*_t}		\right)\,dt\\
		&~+d\left\langle	e^{\gamma( \widehat X^{\pi,c}-\widehat X^{ *})}, 	M^{ *}	 \right\rangle_t \\
		=&~e^{\gamma(  \widehat X^{ *}_t-\widehat X^{\pi,c}_t )}\left(M^{\pi,c}_t	 -\int_0^t \frac{\alpha}{\gamma}e^{\gamma\widehat c_s+\gamma\widehat X^{\pi,c}_s-\theta\gamma\widehat\nu^*_s}\,ds			\right)\,de^{\gamma( \widehat X^{\pi,c}_t - \widehat X^{ *}_t )}\qquad(\textrm{using }\eqref{expression-M})\\
		&~+e^{\gamma( \widehat X^{\pi,c}_t - \widehat X^{*}_t )}\,dM^{*}_t + \frac{1}{\gamma}e^{ \gamma\widehat X^{\pi,c}_t	+Y_t	}\bigg(	 \alpha  e^{\gamma\widehat c_t-\theta\gamma\widehat\nu^*_t-Y_t} - \alpha  e^{\gamma\widehat c^*_t-\theta\gamma\widehat\nu^*_t-Y_t}	\bigg)\,dt\\
		&~+d\left\langle	e^{\gamma( \widehat X^{\pi,c}-\widehat X^{ *})}, 	M^{*}	 \right\rangle_t \\
		=&~\left(M^{\pi,c}_t	 -\int_0^t \frac{\alpha}{\gamma}e^{\gamma\widehat c_s+\gamma\widehat X^{\pi,c}_s-\theta\gamma\widehat\nu^*_s}\,ds			\right)\bigg\{	\gamma (\pi_t-\pi^*_t)h_t-\gamma(c_t-c^*_t)-\frac{\gamma}{2}(\pi^2_t-(\pi^*_t)^2)(\sigma^2_t+(\sigma^0_t)^2)\\
		&~	+\frac{\gamma^2}{2} (\pi_t-\pi^*_t)^2(	\sigma^2_t+(\sigma^0_t)^2	)  \bigg\}\,dt			\\
		&~	+\left(M^{\pi,c}_t	 -\int_0^t \frac{\alpha}{\gamma}e^{\gamma\widehat c_s+\gamma\widehat X^{\pi,c}_s-\theta\gamma\widehat\nu^*_s}\,ds		\right)\bigg\{\gamma\sigma_t(\pi_t-\pi^*_t)\,dW_t+\gamma\sigma^0_t(\pi_t-\pi^*_t)\,dW^0_t\bigg\}\qquad (\textrm{using }\eqref{eq:to-supermartingale-1})\\
		&~+\left(	  	M^{\pi,c}_t	- \int_0^t \frac{\alpha}{\gamma}e^{\gamma\widehat c_s+\gamma\widehat X^{\pi,c}_s-\theta\gamma\widehat\nu^*_s}\,ds				\right)(\breve Z_t\,dW_t+\breve Z^0_t\,dW^0_t)\qquad (\textrm{using }\eqref{martingale-representation}\textrm{ and }\eqref{expression-M})\\
		&~+ \left(	  	M^{\pi,c}_t	- \int_0^t \frac{\alpha}{\gamma}e^{\gamma\widehat c_s+\gamma\widehat X^{\pi,c}_s-\theta\gamma\widehat\nu^*_s}\,ds				\right)\bigg(	 \alpha  e^{\gamma\widehat c_t-\theta\gamma\widehat\nu^*_t-Y_t} - \alpha  e^{\gamma\widehat c^*_t-\theta\gamma\widehat\nu^*_t-Y_t}	\bigg)\,dt\qquad (\textrm{using }\eqref{def:Y-2})\\
		&~+\left(	  	M^{\pi,c}_t	- \int_0^t \frac{\alpha}{\gamma}e^{\gamma\widehat c_s+\gamma\widehat X^{\pi,c}_s-\theta\gamma\widehat\nu^*_s}\,ds				\right)	\bigg\{	\gamma\sigma_t\breve Z_t + \gamma\sigma^0_t\breve Z^0_t				\bigg\}(\pi_t-\pi^*_t)\,dt				\quad 								(\textrm{using }\eqref{martingale-representation},~  \eqref{eq:to-supermartingale-1}\textrm{ and }\eqref{expression-M})\\
		=&~\left(	  	\gamma M^{\pi,c}_t	- \int_0^t  \alpha e^{\gamma\widehat c_s+\gamma\widehat X^{\pi,c}_s-\theta\gamma\widehat\nu^*_s}\,ds				\right)\Bigg\{		\frac{ (1-\gamma)(\sigma^2_t+(\sigma^0_t)^2)}{2}\left(	\frac{ \pi^*_t}{1-\gamma}-\frac{  h_t+\sigma_t \breve Z_t+\sigma^0_t\breve Z^0_t }{	(1-\gamma)(\sigma^2_t+(\sigma^0_t)^2)	}	\right)^2	  	\\
		&~\qquad\qquad\qquad\qquad\qquad\qquad\qquad\qquad-			\frac{ (1-\gamma)(\sigma^2_t+(\sigma^0_t)^2)}{2}\left(	
		   \pi_t - \frac{  h_t+\sigma_t \breve Z_t+\sigma^0_t\breve Z^0_t }{	(1-\gamma)(\sigma^2_t+(\sigma^0_t)^2)} +\frac{\gamma\pi^*_t}{1-\gamma}	\right)^2	\\
		&~        +  c^*_t-\frac{\alpha}{\gamma}e^{\gamma\widehat c^*_t-\theta\gamma\widehat\nu^*_t-Y_t} +\frac{1-\gamma}{\gamma}\{	\alpha e^{-Y_t-\theta\gamma\widehat\nu^*_t}	\}^{\frac{1}{1-\gamma}}	-  c_t+\frac{\alpha}{\gamma} e^{\gamma\widehat c_t-\theta\gamma\widehat\nu^*_t-Y_t}-\frac{1-\gamma}{\gamma}\{	\alpha e^{-Y_t-\theta\gamma\widehat\nu^*_t}	\}^{\frac{1}{1-\gamma}}	\Bigg\}\,dt\\
		&~+ \left(	  	M^{\pi,c}_t	- \int_0^t \frac{\alpha}{\gamma}e^{\gamma\widehat c_s+\gamma\widehat X^{\pi,c}_s-\theta\gamma\widehat\nu^*_s}\,ds				\right)\bigg\{	(	\gamma\sigma_t\pi_t-\gamma\sigma_t\pi^*_t+\breve Z_t	)\,dW_t+(\gamma\sigma^0_t\pi_t-\gamma\sigma^0_t\pi^*_t+\breve Z^0_t		)\,dW^0_t				\bigg\}\\
		:=&~f^{\pi,c}_{1,t}\,dt+f^{\pi,c}_{2,t}\,dW_t+f^{\pi,c}_{3,t}\,dW^0_t.
%	\end{split}
\end{align*}
In the next step, we will verify $\int_0^\cdot f^{\pi,c}_{2,s}\,dW_s+\int_0^\cdot f^{\pi,c}_{3,s}\,dW^0_s$ is a martingale for each $(\pi,c)\in L^\infty\times L^\infty$. 

\textbf{Step 3: verification of martingale properties.} 
Since all coefficients are bounded and $(\pi^*,c^*,\pi,c)\in L^\infty\times L^\infty\times L^\infty\times L^\infty$, it is sufficient to verify 
\begin{equation*}
	\begin{split}
%		&~	\mathbb E\left[	\int_0^T\left(	  	M^{\pi,c}_t	- \int_0^t \frac{\alpha}{\gamma}e^{\gamma\widehat c_s+\gamma\widehat X^{\pi,c}_s-\theta\gamma\widehat\nu^*_s}\,ds				\right)^2\bigg\{	(	\gamma\sigma_t\pi_t-\gamma\sigma_t\pi^*_t+\breve Z_t	)^2\,dt+(\gamma\sigma^0_t\pi_t-\gamma\sigma^0_t\pi^*_t+\breve Z^0_t		)^2\,dt	\bigg\}\,dt				\right]	\\
		\mathbb E\left[	\int_0^T\left(	  	M^{\pi,c}_t	- \int_0^t \frac{\alpha}{\gamma}e^{\gamma\widehat c_s+\gamma\widehat X^{\pi,c}_s-\theta\gamma\widehat\nu^*_s}\,ds				\right)^2\bigg\{	1+\breve Z^2_t +(\breve Z^0_t)^2		 	\bigg\}\,dt				\right]	
		<\infty.
	\end{split}
\end{equation*}
From \eqref{relation-Z-M*} in {\bf Step 1}, we have that
\begin{equation*}
	\begin{split}
		&~\mathbb E\left[	\int_0^T\left(	  	M^{\pi,c}_t	- \int_0^t \frac{\alpha}{\gamma}e^{\gamma\widehat c_s+\gamma\widehat X^{\pi,c}_s-\theta\gamma\widehat\nu^*_s}\,ds				\right)^2 \breve Z^2_t   \,dt				\right]\\
		=&~\mathbb E\left[	\int_0^T\left(	  	M^{\pi,c}_t	- \int_0^t \frac{\alpha}{\gamma}e^{\gamma\widehat c_s+\gamma\widehat X^{\pi,c}_s-\theta\gamma\widehat\nu^*_s}\,ds				\right)^2\frac{(M^{*}_t)^2}{ 	  (	M^{*}_t	- \int_0^t \frac{\alpha}{\gamma}e^{\gamma\widehat c^*_s+\gamma\widehat X^{*}_s-\theta\gamma\widehat\nu^*_s}\,ds)^2				 }\mathring Z^2_t   \,dt				\right]\\
		\leq&~\mathbb E\left[	\frac{\sup_{0\leq t\leq T}\left(	  	M^{\pi,c}_t	- \int_0^t \frac{\alpha}{\gamma}e^{\gamma\widehat c_s+\gamma\widehat X^{\pi,c}_s-\theta\gamma\widehat\nu^*_s}\,ds				\right)^2\sup_{0\leq t\leq T}(M^{*}_t)^2}{ \inf_{0\leq t\leq T}	  \left(	M^{*}_t	- \int_0^t \frac{\alpha}{\gamma}e^{\gamma\widehat c^*_s+\gamma\widehat X^{*}_s-\theta\gamma\widehat\nu^*_s}\,ds\right)^2				 }\int_0^T\mathring Z^2_t   \,dt				\right]\\
		:=&~ \mathbb E\left[	\frac{I_{1,T}}{I_{2,T}}\int_0^T\mathring Z^2_t\,dt		\right].
	\end{split}
\end{equation*}
We estimate $I_{1,T}$ and $I_{2,T}$ separately. For the denominator $I_{2,T}$ we have that by the boundedness of all coefficients and $(\pi^*,c^*)$
\begin{equation*}
	\begin{split}
	&~	\inf_{0\leq t\leq T}	  \left(	M^{*}_t	- \int_0^t \frac{\alpha}{\gamma}e^{\gamma\widehat c^*_s+\gamma\widehat X^{*}_s-\theta\gamma\widehat\nu^*_s}\,ds\right)^2	\\
	\geq &~ \frac{1}{\gamma^2}\inf_{0\leq t\leq T}\mathbb E\left[\left. e^{\gamma(\widehat X^*_T-\theta\widehat\mu^*_T)  }	\right|\mathcal F_t	\right]^2\\
	\geq&~C\inf_{0\leq t\leq T}\mathbb E\left[\left.\mathcal E\left(		\int_0^T\gamma\sigma_t\pi^*_t\,dW_t+\int_0^T\left(\gamma\sigma^0_t\pi^*_t-\theta\gamma\mathbb E[ \gamma\sigma_t\pi^*_t|\mathcal F^0_t ]		\right)\,dW^0_t	\right)\right|\mathcal F_t\right]^2\\
	=&~C\inf_{0\leq t\leq T}\mathcal E\left(			\int_0^t\gamma\sigma_s\pi^*_s\,dW_s+\int_0^t\left(\gamma\sigma^0_s\pi^*_s-\theta\gamma\mathbb E[ \gamma\sigma_s\pi^*_s|\mathcal F^0_s ]		\right)\,dW^0_s	\right)^2,
	\end{split}
\end{equation*}
from which using again the boundedness of all coefficients and $(\pi^*,c^*)$ we have that 
\begin{equation*}
	\begin{split}
	&~	\frac{1}{		\inf_{0\leq t\leq T}	  \left(	M^{*}_t	- \int_0^t \frac{\alpha}{\gamma}e^{\gamma\widehat c^*_s+\gamma\widehat X^{*}_s-\theta\gamma\widehat\nu^*_s}\,ds\right)^2		}\\
	\leq&~ \frac{C}{\inf_{0\leq t\leq T}\mathcal E\left(			\int_0^t\gamma\sigma_s\pi^*_s\,dW_s+\int_0^t\left(\gamma\sigma^0_s\pi^*_s-\theta\gamma\mathbb E[ \gamma\sigma_s\pi^*_s|\mathcal F^0_s ]		\right)\,dW^0_s	\right)^2}\\
	\leq&~\sup_{0\leq t\leq T} \frac{C}{ \mathcal E\left(			\int_0^t\gamma\sigma_s\pi^*_s\,dW_s+\int_0^t\left(\gamma\sigma^0_s\pi^*_s-\theta\gamma\mathbb E[ \gamma\sigma_s\pi^*_s|\mathcal F^0_s ]		\right)\,dW^0_s	\right)^2}\\
	\leq&~C\sup_{0\leq t\leq T} \mathcal E\left(-\int_0^t2\gamma\sigma_s\pi^*_s\,dW_s-\int_0^t2\left(\gamma\sigma^0_s\pi^*_s-\theta\gamma\mathbb E[ \gamma\sigma_s\pi^*_s|\mathcal F^0_s ]		\right)\,dW^0_s\right).
	\end{split}
\end{equation*}
For the numerator $I_{1,T}$, it holds that
\begin{equation*}
	\begin{split}
	&~	\sup_{0\leq t\leq T}\left(	  	M^{\pi,c}_t	- \int_0^t \frac{\alpha}{\gamma}e^{\gamma\widehat c_s+\gamma\widehat X^{\pi,c}_s-\theta\gamma\widehat\nu^*_s}\,ds				\right)^2\sup_{0\leq t\leq T}(M^{*}_t)^2\\
	\leq&~\sup_{0\leq t\leq T} e^{2\gamma( \widehat X^{\pi,c}_t - \widehat X^{*}_t )}\left(M^{*}_t	 -\int_0^t \frac{\alpha}{\gamma}e^{\gamma\widehat c^*_s+\gamma\widehat X^{*}_s-\theta\gamma\widehat\nu^*_s}\,ds			\right)^2\sup_{0\leq t\leq T}(M^{*}_t)^2\\
	\leq&~C\sup_{0\leq t\leq T} e^{2\gamma( \widehat X^{\pi,c}_t - \widehat X^{*}_t )}\sup_{0\leq t\leq T} \mathbb E\left[	\left.	 e^{ \gamma( \widehat X^*_T-\theta\widehat\mu^*_T  ) }  +\int_0^T e^{\gamma\widehat c^*_s+\gamma\widehat X^*_s-\theta\gamma\widehat\nu^*_s}\,ds 		\right|\mathcal F_t\right]^4. 
%	\leq&~ C\sup_{0\leq t\leq T} e^{2\gamma\int_0^t\sigma_s(\pi_s-\pi^*_s)\,dW_s+2\gamma\int_0^t\sigma^0_s(\pi_s-\pi^*_s)\,dW^0_s}\sup_{0\leq t\leq T}\mathbb E\bigg[	e^{\int_0^T \pi^*_s\sigma_s\gamma\,dW_s+\int_0^T\gamma\{ \pi^*_s\sigma^0_s-\theta\mathbb E[ \pi^*_s\sigma^0_s|\mathcal F^0_s ]		\}\,dW^0_s  }	\\
%	&~\qquad\qquad\qquad\qquad\qquad\qquad\qquad\qquad\qquad\quad+\int_0^T e^{\int_0^t \pi^*_s\sigma_s\gamma\,dW_s +\int_0^t\gamma\{	\pi^*_s\sigma^0_s-\theta\mathbb E[ \pi^*_s\sigma^0_s|\mathcal F^0_s     ]	\}\,dW^0_s }\,dt\Big|\mathcal F_t	\bigg]^4.
	\end{split}
\end{equation*}
Thus, H\"older's inequality and Doob's maximal inequality imply that
\begin{equation*}
	\begin{split}
		&~\mathbb E\left[	\int_0^T\left(	  	M^{\pi,c}_t	- \int_0^t \frac{\alpha}{\gamma}e^{\gamma\widehat c_s+\gamma\widehat X^{\pi,c}_s-\theta\gamma\widehat\nu^*_s}\,ds				\right)^2 \breve Z^2_t   \,dt				\right]\\
		\leq&~C\mathbb E\left[	\sup_{0\leq t\leq T} \mathcal E\left(-\int_0^t2\gamma\sigma_s\pi^*_s\,dW_s-\int_0^t2\left(\gamma\sigma^0_s\pi^*_s-\theta\gamma\mathbb E[ \gamma\sigma_s\pi^*_s|\mathcal F^0_s ]		\right)\,dW^0_s\right)^{p_1}	\right]^{\frac{1}{p_1}}\mathbb E\left[ \sup_{0\leq t\leq T}e^{2p_2\gamma( \widehat X^{\pi,c}_t - \widehat X^{*}_t )}	\right]^{\frac{1}{p_2}} \\
		&~\times \mathbb E\left[	\sup_{0\leq t\leq T}  \mathbb E\bigg[	e^{\gamma(\widehat X^*_T-\theta\widehat\mu^*_T)  }+\int_0^T e^{\gamma\widehat c^*_t+\gamma\widehat X^*_t-\theta\gamma\widehat\nu^*_t }\,dt\Big|\mathcal F_t	\bigg]^{4p_3}	\right]^{\frac{1}{p_3}}\\
		&~\times \mathbb E\left[\left(	\int_0^T\mathring Z^2_t\,dt		\right)^{p_4}		\right]^{\frac{1}{p_4}}\qquad\qquad\qquad\qquad(\frac{1}{p_1}+\frac{1}{p_2}+\frac{1}{p_3}+\frac{1}{p_4}=1)\\
		\leq&~ C \mathbb E\left[	\sup_{0\leq t\leq T} \mathcal E\left(-\int_0^t2\gamma\sigma_s\pi^*_s\,dW_s-\int_0^t2\left(\gamma\sigma^0_s\pi^*_s-\theta\gamma\mathbb E[ \gamma\sigma_s\pi^*_s|\mathcal F^0_s ]		\right)\,dW^0_s\right)^{p_1}	\right]^{\frac{1}{p_1}}\mathbb E\left[ \sup_{0\leq t\leq T}e^{2p_2\gamma( \widehat X^{\pi,c}_t - \widehat X^{*}_t )}	\right]^{\frac{1}{p_2}}\\
		&~\times \Bigg\{\mathbb E\left[	e^{4p_3\gamma(X^*_T-\theta\widehat\mu^*_T)}	\right]^{\frac{1}{p_3}}	+\mathbb E\left[	\int_0^T e^{4p_3(\gamma\widehat c^*_t+\gamma\widehat X^*_t-\theta\gamma\widehat\nu^*_t)}\,dt 	\right]^{\frac{1}{p_3}}		\Bigg\}  \mathbb E\left[\left(	\int_0^T\mathring Z^2_t\,dt		\right)^{p_4}		\right]^{\frac{1}{p_4}}\\
		<&~\infty,
		\end{split}
\end{equation*}
where $\mathbb E\left[\left(	\int_0^T\mathring Z^2_t\,dt		\right)^{p_4}		\right]<\infty$ is due to $\mathring Z\in H^2_{BMO}$ and the energy inequality; refer to \cite[P.26]{Kazamaki-2006}. Similarly, one also has
\[
	\mathbb E\left[	\int_0^T\left(	  	M^{\pi,c}_t	- \int_0^t \frac{\alpha}{\gamma}e^{\gamma\widehat c_s+\gamma\widehat X^{\pi,c}_s-\theta\gamma\widehat\nu^*_s}\,ds				\right)^2 (1+(\breve Z_t^0)^2)   \,dt				\right]<\infty.
\]

\textbf{Step 4: complete the proof.}  Since $M^{\pi,c}$ is a supermartingale and $\int_0^\cdot f^{\pi,c}_{2,s}\,dW_s+\int_0^\cdot f^{\pi,c}_{3,s}\,dW^0_s$ is a martingale from {\bf Step 3}, it holds that $M^{\pi,c}_\cdot-\int_0^\cdot f^{\pi,c}_{2,s}\,dW_s-\int_0^\cdot f^{\pi,c}_{3,s}\,dW^0_s$ is a supermartingale, i.e. $\int_0^\cdot f^{\pi,c}_{1,t}\,dt$ is a supermartingale for all $(\pi,c)$. It implies that $f^{\pi,c}_1\leq 0$ for all $(\pi,c)$, and $f^{\pi^*,c^*}_1=0$ since $M^*$ is a martingale.
Thus, we have
\[
	\pi^*=\frac{  h+\sigma\breve Z+\sigma^0\breve Z^0		}{\sigma^2+(\sigma^0)^2},\quad   c^*=\alpha^{\frac{1}{1-\gamma}} e^{ -\frac{1}{1-\gamma}(\theta\gamma\widehat\nu^*+Y)  }. 
\]
Define $(Z,Z^0)=(\breve Z-\gamma\sigma\pi^*,\breve Z^0-\gamma\sigma^0\pi^*)$. Then 
\[
		\pi^*=\frac{  h+\sigma  Z+\sigma^0  Z^0		}{(1-\gamma)(\sigma^2+(\sigma^0)^2)}.
\]
 Recall $Y$ defined in \eqref{def:Y-1}. Then $(\widehat X^{*},Y,  Z,Z^0)$ satisfies the FBSDE \eqref{FBSDE-power}.

$\bm{(2)}$ Let $(\widehat X^*,Y,Z,Z^0)$ be a solution to \eqref{FBSDE-power}. By Proposition \ref{prop:BSDE-optim}, together with the probabilistic approach in \cite{CD-2013}, $(\mu^*,\nu^*,\pi^*,c^*)$ is an NE of \eqref{model-MFG-power}, with $\widehat\nu^*=\mathbb E[\widehat{c^*X^*}|\mathcal F^0]$, $\widehat\mu^*$, $\pi^*$ and $c^*$ satisfying the first, the third and the last equality in \eqref{FBSDE-characterization}.  It remains to verify that $\widehat\nu^*$ satisfies the second equality in \eqref{FBSDE-characterization}.   %Let $(\widehat X^*,Y,Z,Z^0)$ be a solution to \eqref{FBSDE-power} with $(Z,Z^0)\in H^2_{BMO}\times H^2_{BMO}$, and let  $(\mu^*,\nu^*,\pi^*,c^*)$ be given in \eqref{FBSDE-characterization}. We will prove $(\mu^*,\nu^*,\pi^*,c^*)$ is an NE of the MFG \eqref{model-MFG-power}. To do so, for each strategy $(\pi,c)\in L^\infty\times L^\infty$, define $$R^{\pi,c}_t=\frac{1}{\gamma}  e^{\gamma\widehat X^{\pi,c}_t+Y_t}+\int_0^t\frac{\alpha}{\gamma}e^{\gamma\widehat c_s+\gamma\widehat X^{\pi,c}_s-\theta\gamma\widehat\nu^*_s}\,ds.$$ 
%We will prove that $R^{\pi,c}$ satisfies the following three items:
%\begin{equation}\label{claim-MOP}
%	\begin{split}
%		1. &\quad 
%	\textrm{for any }(\pi,c), ~R^{\pi,c}\textrm{ is a supermartingale};\\
%		2. &\quad R^{\pi^*,c^*} \textrm{ is a martingale for } (\pi^*,c^*) \textrm{ in }\eqref{FBSDE-characterization};\\
%		3. &\quad R^{\pi,c}_0 \textrm{ is independent of }(\pi,c).
%	\end{split}
%\end{equation}
%If the claim \eqref{claim-MOP} is true, it holds that $$\mathbb E[R^{\pi,c}_T]\leq R_0^{\pi,c}=R_0^{\pi^*,c^*}=\mathbb E[R^{\pi^*,c^*}_T].$$
%Thus,
%$(\pi^*,c^*)$ is optimal given $(\mu^*,\nu^*)$. Moreover, 
Indeed, by the last equality in \eqref{FBSDE-characterization}, it holds that
\begin{equation}\label{eq:hat-nu}
	\begin{split}
		\widehat \nu^*_t=&~\mathbb E[\widehat c^*_t|\mathcal F^0_t] + \mathbb E[\widehat X^*_t|\mathcal F^0_t]\\
		=&~\mathbb E\left[	\frac{\log\alpha}{1-\gamma}	\right] -\mathbb E\left[ \left.\frac{Y_t}{1-\gamma}		\right|\mathcal F^0_t\right]-\mathbb E\left[ \left.\frac{\theta\gamma\widehat\nu^*_t}{1-\gamma}		\right|\mathcal F^0_t\right] + \mathbb E[	\widehat X^*_t|\mathcal F^0_t	].
	\end{split}
\end{equation}
Multiplying $\frac{\theta\gamma}{1-\gamma}$ and taking conditional expectations $\mathbb E[\cdot|\mathcal F^0_t]$ on both sides of \eqref{eq:hat-nu}, we have
\begin{equation*}
	\begin{split}
		\mathbb E\left[\left.	\frac{\theta\gamma}{1-\gamma}\widehat\nu^*_t\right|\mathcal F^0_t	\right] =&~ \mathbb E\left[	\frac{\theta\gamma}{1-\gamma}	\right]\mathbb E\left[	\frac{\log\alpha}{1-\gamma}	\right] - \mathbb E\left[	\frac{\theta\gamma}{1-\gamma}	\right]\mathbb E\left[ \left.\frac{Y_t}{1-\gamma}		\right|\mathcal F^0_t\right]\\
		&~-\mathbb E\left[\frac{\theta\gamma}{1-\gamma}	\right]\mathbb E\left[ \left.\frac{\theta\gamma\widehat\nu^*_t}{1-\gamma}		\right|\mathcal F^0_t\right] + \mathbb E\left[	\frac{\theta\gamma}{1-\gamma}	\right]\mathbb E[	\widehat X^*_t|\mathcal F^0_t	],
	\end{split}
\end{equation*}
which implies that 
\[
	\mathbb E\left[\left. \frac{\theta\gamma\widehat\nu^*_t}{1-\gamma}  \right|\mathcal F^0	\right] = \frac{\mathbb E\left[	\frac{\theta\gamma}{1-\gamma}\right]}{1+\mathbb E\left[ \frac{\theta\gamma}{1-\gamma}	\right]	} \left\{	  \mathbb E\left[	\frac{\log\alpha}{1-\gamma}	\right] - \mathbb E\left[ \left.\frac{Y_t}{1-\gamma}		\right|\mathcal F^0_t\right]
	+\mathbb E\left[\widehat X^*_t|\mathcal F^0_t	\right]	\right\}. 
\]
Taking the above equality back into \eqref{eq:hat-nu}, we obtain the second equality in \eqref{FBSDE-characterization}.
\end{proof}
The following two remarks show that portfolio games with exponential utility functions and log utility functions are also equivalent to some FBSDEs.
\begin{remark}[MFGs with exponential utility functions]
	If each player uses an exponential utility criterion, then the MFG becomes:
		\begin{equation}\label{model-MFG-exp}
		\left\{ \begin{split}
			1.&~\textrm{Fix }(\mu,\nu)\textrm{ in some suitable space};\\
			2.&~\textrm{Solve the optimization problem: }\\
			&~\mathbb E\left[    \int_0^T-\alpha e^{-\beta(c_s-\theta\nu_s)}	\,ds- e^{-\beta( X_T-\theta\mu_T  )}	\right]\rightarrow\max \textrm{ over }(\pi,c)\\
			&~\textrm{such that } dX_t=\pi_t(h_t\,dt+\sigma_t\,dW_t+\sigma^0_t\,dW^0_t)-c_t\,dt,~X_0=x_{exp};\\
			3.&~\textrm{Search for the fixed point }(\mu_t,\nu_t)=\left(\mathbb E[ X^*_t|\mathcal F^0_t],\mathbb E[c^*_t|\mathcal F^0_t]\right),~t\in[0,T],\\
			&~X^*\textrm{ and } c^*\textrm{ are the optimal wealth and consumption from }2.
		\end{split}\right.
	\end{equation}
	Following the same argument in Theorem \ref{thm:NE-FBSDE-power}, the NE of the MFG \eqref{model-MFG-exp} has a one-to-one correspondence with the following mean field FBSDE
	\begin{equation}\label{FBSDE-exp}
	\left\{	\begin{split}
			dX_t=&~(\pi^*_th_t-c^*_t)\,dt+\pi^*_t\sigma_t\,dW_t+\pi^*_t\sigma^0_t\,dW^0_t\\
			-dY_t=&~\bigg\{\frac{1}{2\beta(\sigma^2_t+(\sigma^0_t)^2)}\left\{	h_t-\beta(\sigma_tZ_t+\sigma^0_tZ^0_t)		\right\}^2-\frac{\beta}{2}(Z^2_t+(Z^0_t)^2)\\ &~-g_tY_t-\theta g_t\mathbb E[c^*_t|\mathcal F^0_t]   
			+\frac{g_t}{\beta}\log\frac{g_t}{\alpha}-\frac{g_t}{\beta}\bigg\}\,dt - Z_t\,dW_t-Z^0_t\,dW^0_t,\\
			X_0=&~x_{exp},\quad Y_T=-\theta\mathbb E[X_T|\mathcal F^0_T],
		\end{split}\right.
	\end{equation}
	where $g_t=\frac{1}{1+T-t}$, $0\leq t\leq T$, and the optimal investment and consumption are given by
	\[
		\pi^*=\frac{h-\beta(\sigma Z+\sigma^0 Z^0)}{\beta g(\sigma^2+(\sigma^0)^2)},\quad c^*=gX+Y+\frac{\theta\mathbb E[gX+Y|\mathcal F^0]}{1-\mathbb E[\theta]} -\log\frac{g}{\alpha}-\frac{\theta\mathbb E\left[\frac{1}{\beta}\log\frac{g}{\alpha}\right]}{1-\mathbb E[\theta]} .
	\]
\end{remark}

\begin{remark}[MFGs with log utility functions]
	If each player uses log utility criterion, then the MFG becomes:
	\begin{equation}\label{model-MFG-log}
		\left\{ \begin{split}
			1.&~\textrm{Fix }\mu\textrm{ in some suitable space};\\
			2.&~\textrm{Solve the optimization problem: }\\
			&~\mathbb E\left[ \int_0^T  \alpha\log (c_tX_t\nu^{-\theta}_t) 		\,dt+   \log \big(X_T\mu^{-\theta}_T\big) 		\right]\rightarrow\max \textrm{ over }(\pi,c)\\
			&~\textrm{such that } dX_t=\pi_tX_t(h_t\,dt+\sigma_t\,dW_t+\sigma^0_t\,dW^0_t)-c_tX_t\,dt,~X_0=x_{log};\\
			3.&~\textrm{Search for the fixed point }(\mu_t,\nu_t)=(\exp\left(\mathbb E[\widehat X^*_t|\mathcal F^0_t]\right),\exp\left(\mathbb E[	\widehat{c^*_tX^*_t}|\mathcal F^0]	\right) ),~t\in[0,T],\\
			&~(X^*,c^*)\textrm{ is the optimal wealth and consumption rate from }2.
		\end{split}\right.
	\end{equation}
Note that $\arg\max_{\pi,c}\mathbb E\left[ \int_0^T  \alpha\log (c_tX_t\nu^{-\theta}_t) 		\,dt+   \log \big(X_T\mu^{-\theta}_T\big) 		\right]=\arg\max_{\pi,c}\mathbb E[\int_0^T\alpha\log(c_tX_t)+\log X_T]$. Thus, the MFG with log utility criteria is decoupled; each player makes her decision by disregarding her competitors. By \cite{Cheridito2011}, the NE of \eqref{model-MFG-log} is given by
\begin{equation}\label{NE-log}
	\left\{\begin{split}
&~ \pi^*_t=\frac{h_t}{\sigma^2_t+(\sigma^0_t)^2},\quad c^*_t=\frac{\alpha}{1+\alpha(T-t)},\\
&~\mu^*_t=\exp\left(		\mathbb E[\log(X^*_t)|\mathcal F^0_t]\right), \quad \nu^*_t=\exp\left(	\mathbb E[\log(c^*_tX^*_t)|\mathcal F^0_t]	\right),
	\end{split}\right.
\end{equation}
where $X^*$ together with some $(Y,Z)$ is the unique solution to the (decoupled) FBSDE 
\begin{equation}\label{FBSDE-log}
	\left\{\begin{split}
		dX^*_t=&~\frac{h_t}{\sigma^2_t+(\sigma^0_t)^2}X^*_t(	h_t\,dt+\sigma_t\,dW_t+\sigma^0_t\,dW^0_t	)-\frac{\alpha}{1+\alpha(T-t)}X_t^*\,dt,\\
		dY_t=&~\bigg\{\frac{h^2_t}{2(\sigma^2_t+(\sigma^0_t)^2)}+\frac{\alpha}{1+\alpha(T-t)}\log\frac{\alpha}{1+\alpha(T-t)}+\frac{\alpha}{1+\alpha(T-t)}\bigg\}\,dt\\
		&~+Z_t\,dW_t+Z^0_t\,dW^0_t,\\
		X_0=&~x_{log},~Y_T=-\theta\mathbb E[\log X^*_T|\mathcal F^0_T].
	\end{split}\right.
\end{equation}
%Let $(\mu',\pi')$ be any other NE of \eqref{model-MFG-log}. Given $\mu'$, by MOP in \cite{HIM-2005}, the optimal response is $\frac{h}{\sigma^2+(\sigma^0)^2}$, which is unique since the log utility function is concave. Thus, $\pi'=\frac{h}{\sigma^2+(\sigma^0)^2}$ and $\mu'_t=\exp\Big(	E[\log X_t|\mathcal F^0_t]\Big)$, $t\in[0,T]$ and $X$ is the unique solution of \eqref{FBSDE-log}. Thus, $(\mu^*,\pi^*)=(\mu',\pi')$ and the NE of \eqref{model-MFG-log} is unique. 
\end{remark}

 \subsection{MFGs and Mean Field BSDEs Are Equivalent}

In this section,  based on Theorem \ref{thm:NE-FBSDE-power} we prove that the wellposedness of the MFG \eqref{model-MFG-power} is equivalent to the wellposedness of the following mean field BSDE 
\begin{equation}\label{equivalent-BSDE}
	\begin{split}
		\widetilde Y_t=&~	\int_t^T\Bigg(\mathcal J_{\widetilde Z,\widetilde Z^0}(s) +(1-\gamma)  \exp\left\{\frac{\log\alpha}{1-\gamma}-\frac{\widetilde Y_s}{1-\gamma}+\frac{\theta\gamma\mathbb E\left[\frac{\widetilde Y_s}{1-\gamma}|\mathcal F^0_s	\right]-\theta\gamma\mathbb E\left[\frac{\log\alpha}{1-\gamma}\right]}{(1-\gamma)\left(1+\mathbb E\left[\frac{\theta\gamma}{1-\gamma}\right]\right)}\right\}	 	\\	
		&~+\theta\gamma\mathbb E\left[	\left. \exp\left\{	\frac{\log\alpha}{1-\gamma}	-\frac{\widetilde Y_s}{1-\gamma}	+\frac{		\theta\gamma\mathbb E\left[\left.\frac{\widetilde Y_s}{1-\gamma}\right|\mathcal F^0_s\right] -\theta\gamma\mathbb E\left[\frac{\log\alpha}{1-\gamma}	\right]			}{(1-\gamma)\left(	1+\mathbb E\left[\frac{\theta\gamma}{1-\gamma}\right]		\right)}		\right\}	\right|\mathcal F^0_s	\right]	\Bigg)	\,ds\\
		&~	-\int_t^T\widetilde Z_s\,dW_s-\int_t^T\widetilde Z^0_s\,dW^0_s,
	\end{split}
\end{equation}
where $\mathcal J_{\widetilde Z,\widetilde Z^0}$ includes all terms with $(\widetilde Z,\widetilde Z^0)$, and the expression of $\mathcal J_{\widetilde Z,\widetilde Z^0}$ is presented in Appendix \ref{sec:J}.
Specifically, the optimal consumption rate can be characterized by the $\widetilde Y$-component and the optimal investment rate can be characterized by the $(\widetilde Z,\widetilde Z^0)$-component. In order to establish this equivalence, by Theorem \ref{thm:NE-FBSDE-power}, it is sufficient to prove that there is a one-to-one correspondence between each solution to \eqref{FBSDE-power} and each solution to \eqref{equivalent-BSDE}. This is done in the following proposition. 
\begin{proposition}\label{prop:FBSDE-BSDE}
	There is a one-to-one correspondence between solutions to the FBSDE \eqref{FBSDE-power} and solutions to the BSDE \eqref{equivalent-BSDE}. Let $(\widehat X,Y,Z,Z^0)$ and $(\widetilde Y,\widetilde Z,\widetilde Z^0)$ be a solution to \eqref{FBSDE-power} and \eqref{equivalent-BSDE}, respectively. The relation is given by
	\begin{equation}\label{relation-XYZ}
		\begin{split}
			\widetilde Y=Y+\theta\gamma\mathbb E[\widehat X|\mathcal F^0],\qquad \widetilde Z=Z,\qquad \widetilde Z^0=Z^0+\theta\gamma\mathbb E\left[	\frac{h+\sigma Z+\sigma^0Z^0}{(1-\gamma)(\sigma^2+(\sigma^0)^2)}\sigma^0\Big|\mathcal F^0	\right]. 
		\end{split}
	\end{equation}
\end{proposition}
\begin{proof}
	Let $(\widehat X,Y,Z,Z^0)$ be a solution to \eqref{FBSDE-power}.	
	From the forward dynamics of \eqref{FBSDE-power}, we get
	\begin{equation*}\label{eq:EX}
	\left\{	\begin{split}
			d\mathbb E[\widehat X_t|\mathcal F^0_t]=&~\mathbb E\left[\left.	\frac{(h_t+\sigma_tZ_t+\sigma^0Z^0_t)h_t}{(1-\gamma) (\sigma^2_t+(\sigma^0_t)^2)  } -\left(	\alpha e^{-Y_t}(\nu^*_t)^{-\theta\gamma}		\right)^{\frac{1}{1-\gamma}} -\frac{(h_t+\sigma_tZ_t+\sigma^0_tZ^0_t)^2}{2(1-\gamma)^2(\sigma^2_t+(\sigma^0_t)^2)}	\right|\mathcal F^0_t	\right]\,dt\\
			&~+\mathbb E\left[\left. \frac{h_t+\sigma_tZ_t+\sigma^0_tZ^0_t}{(1-\gamma)(\sigma^2_t+(\sigma^0_t)^2)}\sigma^0_t\right|\mathcal F^0_t\right]\,dW^0_t,\\
		\mathbb E[\widehat X_0|\mathcal F^0_0]=&\mathbb E[\log x].
		\end{split}\right.
	\end{equation*}
Taking the dynamics of $\mathbb E[\widehat X|\mathcal F^0]$ into the dynamics of $Y$, we get
	\begin{equation*}
	\begin{split}
		 Y_t+\theta\gamma\mathbb E[\widehat X_t|\mathcal F^0_t]=&~-\theta\gamma\int_t^T\mathbb E\bigg[\left.	\frac{(h_s+\sigma_sZ_s+\sigma^0_sZ^0_s)h_s}{(1-\gamma)(\sigma^2_s+(\sigma^0_s)^2)} -( \alpha e^{-Y_s}(\nu^*_s)^{-\theta\gamma} )^{\frac{1}{1-\gamma}}	-\frac{  (h_s+\sigma_sZ_s+\sigma^0_sZ^0_s)^2 }{2(1-\gamma)(\sigma^2_s+(\sigma^0_s)^2)} \right|\mathcal F^0_s	\bigg]\,ds\\
		&~+ \int_t^T  \left\{  (1-\gamma)(\alpha e^{-Y_s}(\nu^*_s)^{-\theta\gamma})^{\frac{1}{1-\gamma}} + \frac{Z^2_s+(Z^0_s)^2}{2}	+ \frac{	\gamma(  h_s+ \sigma_sZ_s+ \sigma^0_sZ^0_s   )^2		}{2(1-\gamma)(  \sigma^2_s+(\sigma^0_s)^2 )}		\right\}\,ds\\
		&~-\theta\gamma\int_t^T \mathbb E\left[\left. \frac{h_s+ \sigma_sZ_s+ \sigma^0_sZ^0_s	}{(1-\gamma)( \sigma^2_s+(\sigma^0_s)^2 )}\sigma^0_s\right|\mathcal F^0_s	\right]\,dW^0_s-\int_t^TZ^0_s\,dW^0_s-\int_t^TZ_s\,dW_s.
	\end{split}
\end{equation*}
Define $(\widetilde Y,\widetilde Z,\widetilde Z^0)$ through \eqref{relation-XYZ}. It holds that
\begin{equation}\label{tilde-YZ-YZ}
	\begin{split}
		\widetilde Y_t=&~-\theta\gamma\int_t^T\mathbb E\bigg[\left.	\frac{(h_s+\sigma_sZ_s+\sigma^0_sZ^0_s)h_s}{(1-\gamma)(\sigma^2_s+(\sigma^0_s)^2)} -( \alpha e^{-Y_s}(\nu^*_s)^{-\theta\gamma} )^{\frac{1}{1-\gamma}}	-\frac{  (h_s+\sigma_sZ_s+\sigma^0_sZ^0_s)^2 }{2(1-\gamma)(\sigma^2_s+(\sigma^0_s)^2)} \right|\mathcal F^0_s	\bigg]\,ds\\
		&~+ \int_t^T  \left\{  (1-\gamma)(\alpha e^{-Y_s}(\nu^*_s)^{-\theta\gamma})^{\frac{1}{1-\gamma}} + \frac{Z^2_s+(Z^0_s)^2}{2}	+ \frac{	\gamma(  h_s+ \sigma_sZ_s+ \sigma^0_sZ^0_s   )^2		}{2(1-\gamma)(  \sigma^2_s+(\sigma^0_s)^2 )}		\right\}\,ds\\
		&~-\int_t^T\widetilde Z_s\,dW_s-\int_t^T\widetilde Z^0_s\,dW^0_s. 
	\end{split}
\end{equation}
From the second equation and the third equation in \eqref{relation-XYZ}, we can solve $Z^0$ in terms of $(\widetilde Z,\widetilde Z^0)$:
\begin{equation}\label{Z0-in-ZZ0-tilde}
	\begin{split}
		Z^0=\widetilde Z^0-\frac{	\theta\gamma \mathbb E\left[	  \left.		\frac{		\sigma^0( h+\sigma\widetilde Z+\sigma^0\widetilde Z^0		)	}{(1-\gamma)(\sigma^2+(\sigma^0)^2)}	\right|\mathcal F^0		\right]		}{	1+\mathbb E\left[	 \frac{\theta\gamma(\sigma^0)^2}{(1-\gamma)(\sigma^2+(\sigma^0)^2)}|\mathcal F^0		\right]				}.
	\end{split}
\end{equation}
Taking \eqref{Z0-in-ZZ0-tilde} into \eqref{tilde-YZ-YZ} and by straightforward calculation we get 
\begin{equation}\label{tildeYZ-Y-tildeZ}
	\begin{split}
		\widetilde Y_t=&~ \int_t^T  \bigg\{ \mathcal J_{\widetilde Z,\widetilde Z^0}(s)  + \theta\gamma\mathbb E\bigg[\left.  ( \alpha e^{-Y_s}(\nu^*_s)^{-\theta\gamma} )^{\frac{1}{1-\gamma}}	 \right|\mathcal F^0_s	\bigg]+ (1-\gamma)(\alpha e^{-Y_s}(\nu^*_s)^{-\theta\gamma})^{\frac{1}{1-\gamma}}   	\bigg\}\,ds\\
		&~-\int_t^T\widetilde Z_s\,dW_s-\int_t^T\widetilde Z^0_s\,dW^0_s. 
	\end{split}
\end{equation}
From the second equation in \eqref{FBSDE-characterization} we get
\begin{equation}\label{exp-Y-nu}
	e^{-Y}(\nu^*)^{-\theta\gamma} = \exp\left\{	-\widetilde Y-\frac{\theta\gamma\mathbb E\left[\frac{\log\alpha}{1-\gamma}\right]}{1+\mathbb E\left[ \frac{\theta\gamma}{1-\gamma} \right]}+\frac{  \theta\gamma\mathbb E\left[ \frac{\widetilde Y}{1-\gamma}\Big|\mathcal F^0		\right]  }{	1+\mathbb E\left[\frac{\theta\gamma}{1-\gamma}\right]	} 	\right\}.
\end{equation}
Taking \eqref{exp-Y-nu} into \eqref{tildeYZ-Y-tildeZ}, we obtain \eqref{equivalent-BSDE}. Thus, for each solution to \eqref{FBSDE-power}, we have found a corresponding solution to \eqref{equivalent-BSDE}. 

Let $(\widetilde Y,\widetilde Z,\widetilde Z^0)$ be a solution to \eqref{equivalent-BSDE}. Define $Z=\widetilde Z$ and $Z^0$ by \eqref{Z0-in-ZZ0-tilde}. With $(Z,Z^0)$, define $\pi^*$ by the third equality in \eqref{FBSDE-characterization}. With \eqref{exp-Y-nu}, define $c^*$ by the fourth equality in \eqref{FBSDE-characterization}. Let $\widehat X$ be the unique solution to the forward SDE in \eqref{FBSDE-power}, in terms of the well-defined $\pi^*$ and $c^*$. Define $Y:=\widetilde Y-\theta\gamma\mathbb E[\widehat X|\mathcal F^0]$. One can check that $(\widehat X,Y,Z,Z^0)$ satisfies the FBSDE \eqref{FBSDE-power}. Thus, for each solution to \eqref{equivalent-BSDE}, we have found one corresponding solution to \eqref{FBSDE-power}.
\end{proof}
 Proposition \ref{prop:FBSDE-BSDE} and Theorem \ref{thm:NE-FBSDE-power} together yield the following one-to-one correspondence between each solution to \eqref{equivalent-BSDE} and each NE of \eqref{model-MFG-power}. Moreover, in Section \ref{sec:closed-form-NE} we will use such correspondence to prove that the NE of \eqref{model-MFG-power} is unique.
 \begin{theorem}\label{thm:NE-BSDE}
There is a one-to-one correspondence between each NE of \eqref{model-MFG-power} and each solution to \eqref{equivalent-BSDE}. The relation is given by
\begin{equation}\label{pi-in-tilde-ZZ0}
	\pi^*=\frac{		h+\sigma\widetilde Z +\sigma^0\widetilde Z^0-\frac{	\theta\gamma\sigma^0\mathbb E\left[\left.	\frac{\sigma^0(h+\sigma\widetilde Z+\sigma^0\widetilde Z^0)}{(1-\gamma)(\sigma^2+(\sigma^0)^2)}\right|\mathcal F^0		\right]		}{1+\mathbb E\left[\left.	\frac{\theta\gamma(\sigma^0)^2}{(1-\gamma)(\sigma^2+(\sigma^0)^2)}\right|\mathcal F^0		\right]}	}{(1-\gamma)(\sigma^2+(\sigma^0)^2)}
\end{equation}
and
\begin{equation}\label{c-in-tilde-Y}
	c^*=\exp\left( \frac{\log\alpha}{1-
	\gamma}	-\frac{\widetilde Y}{1-\gamma}-\frac{\theta\gamma\mathbb E\left[\frac{\log\alpha}{1-\gamma}	\right]}{(1-\gamma)\left(1+\mathbb E\left[\frac{\theta\gamma}{1-\gamma}	\right]\right)}		 +\frac{\theta\gamma\mathbb E\left[ \left. \frac{\widetilde Y}{1-\gamma}\right|\mathcal F^0\right]}{(1-\gamma)\left(1+\mathbb E\left[\frac{\theta\gamma}{1-\gamma}\right]\right)}   \right).	 
\end{equation}
%Moreover, there exists at most one NE of \eqref{model-MFG-power} such that $(\pi^*,c^*)\in L^\infty\times L^\infty$.
 \end{theorem}
 \begin{proof}
 	The equalities \eqref{pi-in-tilde-ZZ0} and \eqref{c-in-tilde-Y} are given by \eqref{FBSDE-characterization} and \eqref{relation-XYZ}. %It remains to prove uniqueness. This is done by proving the uniqueness result of the BSDE \eqref{equivalent-BSDE}. Let $(\widetilde Y,\widetilde Z,\widetilde Z^0)$ and $(\widetilde Y',\widetilde Z',\widetilde Z^{0'})$ be two solutions to \eqref{equivalent-BSDE} in $L^\infty\times L^\infty\times L^\infty$. Note that $|e^y-e^{x}|\leq e^{|x|\vee|y|}|x-y|$. Then $(\widetilde Y-\widetilde Y',\widetilde Z-\widetilde Z',\widetilde Z^0-\widetilde Z^{0'})$ satisfies a mean field BSDE with Lipschitz coefficients. Thus, uniqueness follows from standard estimate. 
 \end{proof}
\section{The Unique NE in Closed Form under Additional Assumptions}\label{sec:closed-form-NE}
 Theorem \ref{thm:NE-BSDE} implies that solving the MFG \eqref{model-MFG-power} is equivalent to solving the BSDE \eqref{equivalent-BSDE}. However, it is difficult to solve the BSDE in general, due to the mixture of quadratic growth of $(Z,Z^0)$, conditional mean field terms of $(Z,Z^0)$, and exponential functions of $Y$. Therefore, we leave the general case to future study. In this section, we will solve the BSDE \eqref{equivalent-BSDE} and the MFG \eqref{model-MFG-power} under the following additional assumption.
 
\textbf{Assumption 2.} 	The return rate $h$ and the volatilities $(\sigma,\sigma^0)$ have continuous trajectories and are measurable w.r.t. $\mathcal A$ at each time $t\in[0,T]$.

The following theorem shows the closed form solution to the BSDE \eqref{equivalent-BSDE} as well as the NE of the MFG \eqref{model-MFG-power} under \textbf{Assumption 2}.
\begin{theorem}\label{prop:explicit-Zt}
	Under \textbf{Assumption 1} and \textbf{Assumption 2}, the BSDE \eqref{equivalent-BSDE} admits a unique solution $(\widetilde Y,\widetilde Z,\widetilde Z^0)\in L^\infty\times L^\infty\times L^\infty$, and the MFG \eqref{model-MFG-power} has a unique NE.
	
For each $t\in[0,T]$, define the following quantities:
\[
\phi_t=\mathbb E\left[  \frac{h_t\sigma^0_t}{(1-\gamma)(\sigma^2_t+(\sigma^0_t)^2)}	\right],\quad \psi_t=\mathbb E\left[	 \frac{(\sigma^0_t)^2\theta\gamma}{(1-
	\gamma)(\sigma^2_t+(\sigma^0_t)^2)}	\right],
\]
\begin{equation*}
	\begin{split}
		A_t=&~-\frac{\gamma}{ 2(1-\gamma)(\sigma^2_t+(\sigma^0_t)^2)    } \left(	h_t-\frac{\theta\gamma\sigma^0_t\phi_t}{1+\psi_t}		\right)^2-\frac{ \phi^2_t\theta^2\gamma^2		}{2(1+\psi_t)^2}   +\theta\gamma\mathbb E\left[		\frac{ h^2_t-\frac{\theta\gamma\sigma^0_th_t\phi_t}{1+\psi_t}  }{(1-\gamma) (\sigma^2_t+(\sigma^0_t)^2) }			\right]\\
		&~-\frac{\theta\gamma}{2}\mathbb E\left[	  \frac{		\left(		h_t-\frac{\theta\gamma\sigma^0_t\phi_t}{1+\psi_t}		\right)^2	}{(1-\gamma)^2(\sigma^2_t+(\sigma^0_t)^2)}					\right],
	\end{split}
\end{equation*}
and
\[
B_t=\frac{\theta\gamma}{1-\gamma}\frac{\mathbb E\left[	\frac{A_t}{1-\gamma}	\right]}{1+\mathbb E\left[	\frac{\theta\gamma}{1-\gamma}	\right]}-\frac{A_t}{1-\gamma},\quad D= \exp\left\{ \frac{\log\alpha}{1-\gamma}  -\frac{\theta\gamma\mathbb E\left[ \frac{\log\alpha}{1-\gamma}	\right]}{(1-\gamma)\left(1+\mathbb E\left[ \frac{\theta\gamma}{1-\gamma}	\right]	\right)}\right\}.
\]
The unique solution to the BSDE \eqref{equivalent-BSDE} has the following closed form expression:
	\begin{equation}\label{explicit-Y}
	\left\{	\begin{split}
	\widetilde Y_t=&~-\theta\gamma\mathbb E[\log D]-(1-\gamma)\log D \\
	&~+\theta\gamma\mathbb E\left[  \log\left(	\exp\left(	\int_t^T B_s\,ds	\right) + D\int_t^T\exp\left(	\int_t^s B_r\,dr	\right)\,ds		\right)		\right]\\
	&~+(1-\gamma) \log\left(	\exp\left(	\int_t^T B_s\,ds	\right)+D\int_t^T\exp\left(	\int_t^s B_r\,dr	\right)\,ds		\right)+\log\alpha,\\
	\widetilde Z_t=&~\widetilde Z^0_t=0,\quad t\in[0,T].
	\end{split}\right.
\end{equation}
The unique optimal investment rate and optimal consumption rate have the following closed form expressions: 
\begin{equation}\label{optimal-investment-MFG}
	\pi^*_t = \frac{h_t}{(1-\gamma)(\sigma^2_t+(\sigma^0_t)^2)} -\frac{\theta\gamma\sigma^0_t \phi_t}{(1-\gamma)(\sigma^2_t+(\sigma^0_t)^2)\left(1+\psi_t		\right)} ,\quad t\in[0,T],
\end{equation}
%Let 
%\begin{equation}\label{eq:A}
%	\begin{split}
%		A=\frac{\log\alpha}{1-\gamma}-\frac{\theta\gamma}{(1-\gamma)\left(	1+\mathbb E\left[\frac{\theta\gamma}{1-\gamma}	\right]	\right)} \mathbb E\left[	\frac{\log\alpha}{1-\gamma}	\right]
%	\end{split}
%\end{equation}
%and
%\begin{equation}\label{eq:B}
%	\begin{split}
%		B_t=&~\frac{\theta\gamma}{(1-\gamma)    \left(1+\mathbb E\left[	 \frac{\theta\gamma}{1-\gamma}\right]	 	\right)    } \int_t^T \bigg(	  \mathbb E\left[	\frac{	\gamma(h_s-\theta\gamma\sigma^0_s\mathbb E[\sigma^0_s\pi^*_s])^2				}{ 2(1-\gamma)^2(\sigma^2_s+(\sigma^0_s)^2)  }	\right]		+\mathbb E\left[ \frac{\theta^2\gamma^2}{2(1-\gamma)}\right](\mathbb E[\sigma^0_s\pi^*_s])^2\\
%		&~\qquad\qquad\qquad\qquad \qquad\qquad 	 -\mathbb E\left[ \frac{\theta\gamma}{1-\gamma}		\right]	\mathbb E[h_s\pi^*_s]	+\mathbb E\left[\frac{\theta\gamma}{2(1-\gamma)}	\right]\mathbb E[ (\sigma^2_s+(\sigma^0_s)^2)(\pi^*_s)^2 ]		\bigg)\,ds\\
%		&~-\frac{1}{1-\gamma}\int_t^T\bigg(	  \frac{\gamma(h_s-\theta\gamma\sigma^0_s\mathbb E[\sigma^0_s\pi^*_s])^2}{2(1-\gamma)(\sigma^2_s+(\sigma^0_s)^2)}  +\frac{\theta^2\gamma^2}{2}(\mathbb E[\sigma^0_s\pi^*_s])^2  -\theta\gamma\mathbb E[h_s\pi^*_s]\\
%		&~\qquad \qquad \qquad  +\frac{\theta\gamma}{2}\mathbb E[  (\sigma^2_s+(\sigma^0_s)^2)(\pi^*_s)^2   ]		\bigg)\,ds,\qquad t\in[0,T].
%	\end{split}
%\end{equation}
respectively,
\begin{equation}\label{optimal-consumption-MFG}
	c^*_t=\frac{ D\exp\left\{-\int_t^TB_r\,dr\right\}		}{1+D\int_t^T \exp\left\{-\int_s^TB_r\,dr\right\}\,ds  },\qquad t\in[0,T].
\end{equation}
\end{theorem}
%Moreover, \eqref{explicit-Z} is unique among bounded solutions.
\begin{proof}
	Our goal is to construct a solution to \eqref{equivalent-BSDE}, such that it does not depend on the Brownian path. If the solution does not depend on the Brownian path, by \textbf{Assumption 2}, the BSDE \eqref{tilde-YZ-YZ} implies 
	\begin{equation}\label{tilde-YZ-YZ-ass2}
		\begin{split}
			\widetilde Y_t=&~-\theta\gamma\int_t^T\mathbb E\left[ \pi^*_sh_s -c^*_s	-\frac{1}{2}(\pi^*_s)^2(\sigma^2_s+(\sigma^0_s)^2)  \right]\,ds\\
			&~+ \int_t^T  \left\{  (1-\gamma)\alpha^{\frac{1}{1-\gamma}} e^{-\frac{\widetilde Y_s}{1-\gamma}-\frac{\theta\gamma}{1-\gamma}\mathbb E[\log c^*_s]} + \frac{Z^2_s+(Z^0_s)^2}{2}	+ \frac{	\gamma(  h_s+ \sigma_sZ_s+ \sigma^0_sZ^0_s   )^2		}{2(1-\gamma)(  \sigma^2_s+(\sigma^0_s)^2 )}		\right\}\,ds\\
			&~-\int_t^T\widetilde Z_s\,dW_s-\int_t^T\widetilde Z^0_s\,dW^0_s. 
		\end{split}
	\end{equation} 
	By the theory of BSDEs, we must have $\widetilde Z=\widetilde Z^0=0$, which together with  \eqref{pi-in-tilde-ZZ0} yields the optimal investment rate \eqref{optimal-investment-MFG}.

Taking $\widetilde Z=\widetilde Z^0=0$ into \eqref{tilde-YZ-YZ-ass2} we have
\begin{equation}\label{SDE-Y-tilde}
	\begin{split}
		d\widetilde Y_t=&~\bigg\{\theta\gamma\mathbb E[\pi^*_th_t] - \theta\gamma\mathbb E[c^*_t]-\frac{\theta\gamma}{2}\mathbb E\left[(\pi^*_t)^2(\sigma^2_t+(\sigma^0_t)^2)\right]-\frac{(Z^0_t)^2}{2}-\frac{\gamma(h_t+\sigma^0_tZ^0_t)^2}{2(1-\gamma)(\sigma^2_t+(\sigma^0_t)^2)}\\
		&~-(1-\gamma)\alpha^{\frac{1}{1-\gamma}} e^{ -\frac{\widetilde Y_t}{1-\gamma}-\frac{\theta\gamma\mathbb E[\log c^*_t]}{1-\gamma} }\bigg\}\,dt.
	\end{split}
\end{equation} 
By $\widetilde Z^0=0$ and the last equality in \eqref{relation-XYZ}, we obtain
\[
Z^0=-\frac{\theta\gamma\mathbb E\left[ \frac{h\sigma^0}{(1-\gamma)(\sigma^2+(\sigma^0)^2)}	\right]}{1+\mathbb E\left[  \frac{\theta\gamma(\sigma^0)^2}{(1-\gamma)(\sigma^2+(\sigma^0)^2)}	\right]}.
\]
Plugging the expression of $Z^0$ into \eqref{SDE-Y-tilde}, by straightforward calculation we have that
\begin{equation}\label{ODE-tildeY-1}
	\widetilde Y'_t= A_t-\theta\gamma\mathbb E[c^*_t] -(1-\gamma)\exp\left\{		\frac{\log\alpha}{1-\gamma} - \frac{  \widetilde Y_t }{1-\gamma}  	-\frac{ \theta\gamma\mathbb E[\log c^*_t] }{1-\gamma}			\right\},
\end{equation}
where $A$ is defined in the statement of the theorem.
Let
\begin{equation}\label{Ytilde-Yhat}
	\widehat Y=\frac{	\widetilde Y-\log\alpha			}{1-\gamma}.
\end{equation}
Equation \eqref{c-in-tilde-Y} implies that
\begin{equation}\label{c*-in-Y-hat}
	c^*=\exp\left\{	  -\widehat Y+\frac{\theta\gamma		}{	(1-\gamma)\left(	1+\mathbb E\left[	\frac{\theta\gamma}{1-\gamma}	\right]		\right)	}\mathbb E[\widehat Y]		\right\}. 
\end{equation}
 Noting that $\widetilde Y'=(1-\gamma)\widehat Y'$ and plugging \eqref{c*-in-Y-hat} into \eqref{ODE-tildeY-1}, we obtain
 \begin{equation}\label{ODE-hatY}
 	\begin{split}
 	\widehat Y'_t=&~ \frac{A_t}{1-\gamma}-\frac{\theta\gamma}{1-\gamma} \mathbb E\left[	 \exp\left\{	-\widehat Y_t+\frac{\theta\gamma}{(1-\gamma)\left(	1+\mathbb E\left[	\frac{\theta\gamma}{1-\gamma}	\right]	\right)}\mathbb E[\widehat Y_t]		\right\}				\right] \\
 	&~ -\exp\left\{	-\widehat Y_t+\frac{\theta\gamma}{(1-\gamma)\left(	1+\mathbb E\left[  \frac{\theta\gamma}{1-\gamma}	\right]	\right)}\mathbb E[\widehat Y_t]			\right\}.
	\end{split}
 \end{equation}
Taking expectations and multiplying both sides of \eqref{ODE-hatY} by $\frac{\theta\gamma}{(1-\gamma)\left(	1+\mathbb E\left[	\frac{\theta\gamma}{1-\gamma}	\right]	\right)}$ , we obtain
\begin{equation}\label{ODE-EhatY}
	\begin{split}
		\frac{\theta\gamma}{(1-\gamma)\left(	1+\mathbb E\left[	\frac{\theta\gamma}{1-\gamma}	\right]	\right)}\mathbb E[\widehat Y_t]'=&~\frac{\theta\gamma}{(1-\gamma)\left(	1+\mathbb E\left[	\frac{\theta\gamma}{1-\gamma}	\right]	\right)}\mathbb E\left[	\frac{A_t}{1-\gamma}	\right] \\
		&~  - \frac{\theta\gamma}{1-\gamma}  \mathbb E\left[	 \exp\left\{	-\widehat Y_t+\frac{\theta\gamma}{(1-\gamma)\left(	1+\mathbb E\left[	\frac{\theta\gamma}{1-\gamma}	\right]	\right)}\mathbb E[\widehat Y_t]		\right\}				\right].
	\end{split}
\end{equation}
Let
\begin{equation}\label{Yring-Yhat}
	\mathring Y= 	\frac{\theta\gamma}{(1-\gamma)\left(	1+\mathbb E\left[	\frac{\theta\gamma}{1-\gamma}	\right]	\right)}\mathbb E[\widehat Y]	 -\widehat Y.
\end{equation}
The difference of \eqref{ODE-hatY} and \eqref{ODE-EhatY} yields an ODE for $\mathring Y$
\begin{equation}\label{ODE-ringY}
	\mathring Y'_t=-\frac{A_t}{1-\gamma}+\frac{\theta\gamma}{(1-\gamma)\left(	1+\mathbb E\left[\frac{\theta\gamma}{1-\gamma}	  \right]	\right)}\mathbb E\left[	\frac{A_t}{1-\gamma}	\right] + \exp\left\{	\mathring Y_t	\right\}.
\end{equation}
Let $\breve Y=\exp(\mathring Y)$. Then, $\breve Y$ satisfies the following Riccati equation
\begin{equation}\label{ODE-breveY}
	\breve Y'_t=\left(		-\frac{A_t}{1-\gamma}+\frac{\theta\gamma}{(1-\gamma)\left(	1+\mathbb E\left[\frac{\theta\gamma}{1-\gamma}	  \right]	\right)}\mathbb E\left[	\frac{A_t}{1-\gamma}	\right]			\right)\breve Y_t + \breve Y^2_t,
\end{equation}
with terminal condition $\breve Y_T=\exp\left\{	-\frac{\theta\gamma\mathbb E\left[ \frac{\log\alpha}{1-\gamma} \right]}{(1-\gamma)\left(	1+\mathbb E\left[\frac{\theta\gamma}{1-\gamma}\right]		\right)}+\frac{\log\alpha}{1-\gamma}		\right\}$. The unique solution to the Riccati equation \eqref{ODE-breveY} is
\begin{equation}\label{Ybreve}
	\breve Y_t=D\left\{		\exp\left(		\int_t^T B_s\,ds			\right) +D\int_t^T\exp\left(		\int_t^s B_r\,dr		\right)\,ds			\right\}^{-1},
\end{equation}
where $B$ and $D$ are defined in the statement of the theorem. By \eqref{c*-in-Y-hat} and the definition of $\mathring Y$ and $\breve Y$, it holds that $c^*=\breve Y$.

Finally, we obtain the closed form expression for $\widetilde Y$ by \eqref{Ytilde-Yhat}, \eqref{Yring-Yhat} and \eqref{Ybreve}.

By now, we have constructed one solution to the BSDE \eqref{equivalent-BSDE} and the MFG \eqref{model-MFG-power}. 
It remains to prove the uniqueness result. Let $(\mu^*,\nu^*,\pi^*,c^*)$ and $(\mu^{*'},\nu^{*'},\pi^{*'},c^{*'})$ be two solutions of the MFG \eqref{model-MFG-power}, where the optimal investment rates are in $L^\infty$. Let $(\widetilde Y,\widetilde Z,\widetilde Z^0)$ and $(\widetilde Y',\widetilde Z',\widetilde Z^{0'})$ be two corresponding solutions to \eqref{equivalent-BSDE}.
Under \textbf{Assumption 2}, the driver of the BSDE \eqref{equivalent-BSDE} does not depend on $W$. Thus, it holds that $\widetilde Z=\widetilde Z'=0$. By \eqref{pi-in-tilde-ZZ0}, we have $\widetilde Z^0,\widetilde Z^{0'}\in L^\infty$. Note that $|e^y-e^{x}|\leq e^{|x|\vee|y|}|x-y|$. Then $(\widetilde Y-\widetilde Y',\widetilde Z-\widetilde Z',\widetilde Z^0-\widetilde Z^{0'})$ satisfies a mean field BSDE with Lipschitz coefficients. Thus, the uniqueness result for the BSDE \eqref{equivalent-BSDE} follows from standard estimate, and the uniqueness result for the MFG \eqref{model-MFG-power} follows from Theorem \ref{thm:NE-FBSDE-power} and Theorem \ref{thm:NE-BSDE}. 
\end{proof}

\begin{remark}
	{ The monotonicity of $\pi^*$ and $c^*$ w.r.t. $\theta$ and $\gamma$ were examined in \cite{LS-2020,LZ-2019}.	This remark investigates the monotonicity of $\pi^*$ and $c^*$ w.r.t. market parameters $h$, $\sigma$ and $\sigma^0$. The same as \cite{LS-2020}, we assume that $h>0$, $\sigma\geq0$ and $\sigma^0\geq0$, which imply that $\phi>0$ and $1+\psi> 0$. When taking derivative w.r.t. some parameter, we assume this parameter is a constant. }

{(1)	Monotonicity of $\pi^*$ w.r.t. the return rate $h$. Direct computation implies that
\[
	\frac{\partial \pi^*}{\partial h} = \frac{1}{(1-\gamma)(  \sigma^2+(\sigma^0)^2 )}>0.
\]		
Thus, the representative player would invest more as the \textit{individual} return rate $h$ increases; it is consistent with intuition and Merton's result \cite{Merton1971}. }

{The dependence of $\pi^*$ on the population's return rate is involved with other population parameters. However, if $h$ is uncorrelated with other population parameters, it holds that
\begin{equation*}
		\frac{  \partial\pi^*}{\partial \bar h} = -\frac{\theta\gamma\sigma^0  \mathbb E\left[  \frac{\sigma^0}{    (1-\gamma) (\sigma^2+(\sigma^0)^2)  }\right] }{(1-\gamma)(\sigma^2+(\sigma^0)^2)1+\psi},\quad \textrm{ which is }\left\{\begin{split}
	 <0, \quad\textrm{ if }\gamma>0,\\
	 >0, \quad\textrm{if }\gamma<0,
	\end{split}\right.
\end{equation*}
where $\bar h$ denotes the average return rate of the population. Thus, when the relative risk aversion is smaller than $1$ (that is, $\gamma>0$), the representative player would invest less if the average return rate of  the population increases. }

{ (2) Monotonicity of $\pi^*$ w.r.t. $\sigma^0$. Direct computation yields
	\[
		\frac{\partial\pi^*}{\partial\sigma^0} = -\frac{2h\sigma^0    }{(1-\gamma)(\sigma^2+(\sigma^0)^2)^2} - \frac{\theta\gamma(\sigma^2-(\sigma^0)^2)}{(1-\gamma)(\sigma^2+(\sigma^0)^2)^2}\frac{\phi}{1+\psi}. 
	\]
Define two thresholds as follows
\[
	\overline\sigma^0=	\frac{	\left\{   h+\sqrt{h^2+\frac{		\theta^2\gamma^2\sigma^2\phi^2	}{(1+\psi)^2}			}	\right\} (1+\psi)	}{	\theta\gamma\phi	}	,\qquad \underline\sigma^0 =	\frac{	\left\{   h-\sqrt{h^2+\frac{		\theta^2\gamma^2\sigma^2\phi^2	}{(1+\psi)^2}			}	\right\} (1+\psi)	}{	\theta\gamma\phi	}.
\]
By the assumption $\sigma^0\geq 0$, it holds that
\[
	\frac{\partial \pi^*}{\partial \sigma^0}<0 \textrm{ if }\sigma^0\in ( 0,\overline\sigma^0\vee \underline\sigma^0  ),\qquad\frac{\partial \pi^*}{\partial \sigma^0}>0\textrm{ if }\sigma^0\in(\overline\sigma^0\vee\underline\sigma^0,\infty).
\]
When there is no competition, i.e. $\theta=0$, $\frac{\partial\pi^*}{\partial\sigma^0}=-\frac{2h\sigma^0}{(1-\gamma)(\sigma^2+(\sigma^0)^2)^2}<0$. 
}

{
When the representative player is less competitive, i.e. $\theta$ is small, then $\overline\sigma^0\vee \underline\sigma^0$ is large so that the volatility is more likely to be located in $(0,\overline\sigma^0\vee \underline\sigma^0)$. Thus, the representative player in the MFG behaves in a similar manner as the one in Merton's problem: the individual investment rate is decreasing w.r.t. the volatility.         }

{
When the representative player is more aggressive, i.e. $\theta$ is large, then $\overline\sigma^0\vee \underline\sigma^0$ takes a relatively smaller value, such that the monotonicity of $\pi^*$  w.r.t. $\sigma^0$ is determined by the threshold $\overline\sigma^0\vee \underline\sigma^0$: when the volatility $\sigma^0$ is larger than $\overline\sigma^0\vee \underline\sigma^0$, the representative player tends to invest more into the risky asset if $\sigma^0$ is larger. The intuitive reason is that the more competitive player is willing to take more risks to expect more returns.   }

{The monotonicity of $\pi^*$ w.r.t. the population's volatility is not tractable. Similar explanations are applicable to the monotonicity w.r.t. $\sigma$.    }

{(3) The monotonicity of $c^*$ w.r.t. $h$, $\sigma$ and $\sigma^0$ is not as tractable as that for $\pi^*$. This is because $\frac{\partial c^*}{\partial \kappa}$ is highly nonlinear in $\kappa$, with $\kappa=h,\sigma,\sigma^0$. %Monotonicity of $\pi^*$ in terms of $\sigma$. Direct computation yields 
%\begin{equation*}
%		\frac{	\partial\pi^*	}{\partial\sigma} = \frac{	2\sigma		}{(1-\gamma)(\sigma^2+(\sigma^0)^2)} \left\{	-h + \frac{	 \theta\gamma\sigma^0 \phi					}{1+\psi}			\right\}, \textrm{ which is }	\left\{\begin{split}
% 			>0,\qquad \textrm{if }h<\frac{\theta\gamma\sigma^0\phi}{1+\psi},\\
% 			<0, \qquad \textrm{if }h>\frac{\theta\gamma\sigma^0\phi}{1+\psi}.
%	\end{split}\right.
%\end{equation*}
}

\end{remark}

As a corollary, when all coefficients become time-independent, we recover the MFG in \cite{LS-2020}. Furthermore, we conclude that the strong NE in \cite{LS-2020} is unique in the essentially bounded space.

\begin{corollary}
	Let \textbf{Assumption 1}  and  \textbf{Assumption 2} hold, and the return rate $h$ and the volatilities $(\sigma,\sigma^0)$ be time-independent. Then the optimal investment rate is
	\begin{equation}\label{eq:constant-strategy-power}
		\pi^{*}=\frac{h}{(1-\gamma)(\sigma^2+(\sigma^0)^2	)}	-\frac{\theta\gamma\sigma^0}{(1-\gamma)(\sigma^2+(\sigma^0)^2)}\frac{\phi}{1+\psi},
	\end{equation}
%Let
%\[
%	\phi=\mathbb E\left[  \frac{h\sigma^0}{(1-\gamma)(\sigma^2+(\sigma^0)^2)}	\right],\quad \psi=\mathbb E\left[	 \frac{(\sigma^0)^2\theta\gamma}{(1-
%		\gamma)(\sigma^2+(\sigma^0)^2)}	\right],
%\]
%\begin{equation*}
%	\begin{split}
%		\rho=&~\frac{\gamma}{ 2(1-\gamma)(\sigma^2+(\sigma^0)^2)    } \left(	h-\frac{\theta\gamma\sigma^0\phi}{1+\psi}		\right)^2+\frac{ \phi^2\theta^2\gamma^2		}{2(1+\psi)^2}   -\theta\gamma\mathbb E\left[		\frac{ h^2-\frac{\theta\gamma\sigma^0h\phi}{1+\psi}  }{(1-\gamma) (\sigma^2+(\sigma^0)^2) }			\right]\\
%		&~+\frac{\theta\gamma}{2}\mathbb E\left[	  \frac{		\left(		h-\frac{\theta\gamma\sigma^0\phi}{1+\psi}		\right)^2	}{(1-\gamma)^2(\sigma^2+(\sigma^0)^2)}					\right],
%	\end{split}
%\end{equation*}
%\[
%	\beta=\frac{\theta\gamma}{1-\gamma}\frac{\mathbb E\left[	\frac{\rho}{1-\gamma}	\right]}{1+\mathbb E\left[	\frac{\theta\gamma}{1-\gamma}	\right]}-\frac{\rho}{1-\gamma},\quad \lambda=\alpha^{\frac{1}{1-\gamma}} e^{-\frac{\theta\gamma\mathbb E\left[ \frac{\log\alpha}{1-\gamma}	\right]}{(1-\gamma)\left(1+\mathbb E\left[ \frac{\theta\gamma}{1-\gamma}	\right]	\right)}}.
%\]
and the optimal consumption rate is
\begin{equation}\label{optimal-consumption-constant}
	c^*_t=\left\{	\begin{split}
		&~\left\{-\frac{1}{B}+\left(	\frac{1}{D}+\frac{1}{B}	\right)e^{B(T-t)}\right\}^{-1},\quad \textrm{if }B\neq 0,\\
		&~\frac{1}{T-t+\frac{1}{D}},\qquad\qquad\qquad \qquad \quad~~  \textrm{if }B=0,	
	\end{split}\right.
\end{equation} 
where $B$ and $D$ are defined in the statement of Theorem \ref{prop:explicit-Zt}, when all market parameters are time-independent. The optimal response $(\pi^*,c^*)$ is unique in $L^\infty\times L^\infty$.
Furthermore, $(\pi^*,c^*)$ given in \eqref{eq:constant-strategy-power} and \eqref{optimal-consumption-constant} is  identical to \cite[Theorem 3.2]{LS-2020}.
\end{corollary}

\section{Conclusion}  
In this paper we study mean field portfolio games with consumption. By MOP and DPP, we establish a one-to-one correspondence between each NE and each solution to some FBSDE. The FBSDE is further proved to be equivalent to some BSDE.  Such equivalence is of vital importance to prove that there exists a unique NE in the essentially bounded space. 
When the market parameters do not depend on the Brownian paths, we get the NE in closed form. Moreover, when the market parameters become time-independent, we recover the model in \cite{LS-2020}, and conclude that the strong NE obtained in \cite{LS-2020} is unique in the essentially bounded space, not only in the space of all strong ones.  
\begin{appendix}

\section{Reverse H\"older Inequality}\label{app:reverse}
For some $p>1$, we say a stochastic process $D$ satisfies reverse H\"older inequality $R_p$ if there exists a constant $C$ such that for each $[0,T]$-valued stopping time $\tau$  it holds that
\[
	\mathbb E\left[	\left|\frac{D_T}{D_\tau}		\right|^p\Big|\mathcal F_\tau	\right]\leq C.
\]
 Let $\Theta$ be a stochastic process and $B$ be a Brownian motion. Define \\[-3mm] 
\[
	\mathcal E_t(\Theta)=\mathcal E\left( \int_0^t\Theta_s\,dB_s   \right).
\]
The following result is from \cite[Theorem 3.4]{Kazamaki-2006}.
%Let $\Phi$ be a function defined on $(1,\infty)$\\[-3mm]
%\[
%	\Phi(x)=\left\{  1+\frac{1}{x^2}\log\frac{2x-1}{2(x-1)}			 \right\}^{\frac{1}{2}}-1,
%\]
%which is a continuous decreasing function satisfying $\lim_{x\searrow 1}\Phi(x)=\infty$ and $\lim_{x\nearrow\infty}\Phi(x)=0$. Let $p_{\Theta}$ be the constant such that $\Phi(p_{\Theta})=\|\Theta\|_{BMO}$. Then we have the following reverse H\"older's inequality. 
\begin{lemma}\label{lem:reverse}
	Let $\mathcal E(\Theta)$ be a uniformly integrable martingale. Then $\Theta\in H^2_{BMO}$ if and only if $\mathcal E(\Theta)$ satisfies $R_p$.
\end{lemma}

\section{The Expression of $\mathcal J_{\widetilde Z,\widetilde Z^0}$ in \eqref{equivalent-BSDE}}\label{sec:J}
The term $\mathcal J_{\widetilde Z,\widetilde Z^0}$ in the driver of \eqref{equivalent-BSDE} includes all terms with $(\widetilde Z,\widetilde Z^0)$. It has the following expression
\begin{equation*}
	\begin{split}
		&~\mathcal J_{\widetilde Z,\widetilde Z^0}(\cdot)\\
		=&~ -\theta\gamma\mathbb E\left[\left.	f^{hh}+f^{\sigma h}\widetilde Z+f^{\sigma^0h}\widetilde Z^0	\right|\mathcal F^0\right]+\theta\gamma\mathbb E\left[\theta\gamma f^{\sigma^0h}|\mathcal F^0	\right]	\frac{\mathbb E\left[	f^{\sigma^0 h}+f^{\sigma^0\sigma}\widetilde Z+f^{\sigma^0\sigma^0}\widetilde Z^0|\mathcal F^0	\right]}{1+\mathbb E\left[ \theta\gamma f^{\sigma^0\sigma^0}|\mathcal F^0	\right]}\\
		&~+\theta\gamma\mathbb E\left[	\left. \frac{1}{2}(1-\gamma)(\sigma^2+(\sigma^0)^2) \left\{	f^h+f^\sigma\widetilde Z+f^{\sigma^0}\widetilde Z^0-\frac{\theta\gamma f^{\sigma^0}\mathbb E[ f^{\sigma^0h}+f^{\sigma^0\sigma}\widetilde Z+f^{\sigma^0\sigma^0}\widetilde Z^0|\mathcal F^0]}{1+\mathbb E[\theta\gamma f^{\sigma^0\sigma^0}|\mathcal F^0]}			\right\}^2	\right|\mathcal F^0\right]\\
		&~+\frac{\widetilde Z^2}{2}+\frac{1}{2}\left\{	\widetilde Z^0-\frac{	\theta\gamma\mathbb E[	f^{\sigma^0h}+f^{\sigma^0\sigma}\widetilde Z+f^{\sigma^0\sigma^0}\widetilde Z^0	|\mathcal F^0]		}{1+\mathbb E[ \theta\gamma f^{\sigma^0\sigma^0}|\mathcal F^0  ]}	\right\}^2\\
		&~+\frac{\gamma(1-\gamma)(\sigma^2+(\sigma^0)^2)}{2}\left\{	f^h+f^\sigma\widetilde Z+f^{\sigma^0}\widetilde Z^0-\frac{\theta\gamma f^{\sigma^0}\mathbb E[	f^{\sigma^0h}+f^{\sigma^0\sigma}\widetilde Z+f^{\sigma^0\sigma^0}\widetilde Z^0|\mathcal F^0	]}{1+\mathbb E[\theta\gamma f^{\sigma^0\sigma^0}|\mathcal F^0]}			\right\}^2,
	\end{split}
\end{equation*}
where for any stochastic processes $a$ and $b$ we denote
\[
	f^a:=\frac{a}{(1-\gamma)(\sigma^2+(\sigma^0)^2)}\quad\textrm{and}\quad	f^{ab}:=\frac{ab}{(1-\gamma)(\sigma^2+(\sigma^0)^2)}.
\]
%\[
%	g:=\mathbb E\left[\left.\frac{	\theta\gamma(\sigma^0)^2	}{(1-\gamma)(\sigma^2+(\sigma^0)^2)    }	\right|\mathcal F^0	\right].
%\]

\end{appendix}

\bibliography{Fu}

\end{document}